\documentclass[review,onefignum,onetabnum]{siamonline171218}

\usepackage{slashbox}
\usepackage{stfloats}
\usepackage{amsmath}
\usepackage{mathbbold}
\usepackage{amssymb}
\usepackage{enumitem}
\setlist[enumerate]{leftmargin=.5in}
\setlist[itemize]{leftmargin=.5in}
\usepackage[algo2e,ruled,vlined,linesnumbered]{algorithm2e}
\crefname{algocf}{alg.}{algs.}
\Crefname{algocf}{Algorithm}{Algorithms}
\usepackage{cleveref}
\usepackage{float}
\usepackage{booktabs}
\usepackage{amsmath,amssymb}
\usepackage{textcomp,marvosym}
\usepackage{nameref,hyperref}

\newtheorem{condition}{Condition}
\newtheorem{remark}{Remark}

\def\dd{\text{d}}

\title{Fixed-budget simulation method for growing cell populations\thanks{\funding{
This work was funded by {the National Natural Science Foundation of China (12471475 \& 11971405 to Da Zhou), the Fundamental Research Funds for the Central Universities (20720240151 \& 20720230023 to Da Zhou), and the Xiamen University Overseas Study Program (to Shaoqing Chen). }}}}

\author{Shaoqing Chen\thanks{Shaoqing Chen and Da Zhou are with the School of Mathematical Sciences, Xiamen University, Xiamen 361005, China.}
\and Zhou Fang\thanks{Zhou Fang is with the Academy of Mathematics and Systems Science, Chinese Academy of
Sciences, Beijing 100190, China. (e-mail: {zhfang@amss.ac.cn})
He was with the Department of Biosystems Science and Engineering, ETH Zurich, Klingelbergstrasse 48, 4056 Basel, Switzerland.}
\and {Zheng Hu\thanks{{Zheng Hu is with Key Laboratory of Quantitative Synthetic Biology, Shenzhen Institute of Synthetic Biology, Shenzhen Institute of Advanced Technology, Chinese Academy of Sciences, Shenzhen 518000, China}}}
\and Da Zhou\footnotemark[2]
}

\begin{document}

\maketitle

\begin{abstract}
    Investigating the dynamics of growing cell populations is crucial for unraveling key biological mechanisms in living organisms, with many important applications in therapeutics and biochemical engineering. 
    Classical agent-based simulation algorithms are often inefficient for these systems because they track each individual cell, making them impractical for fast (or even exponentially) growing cell populations.
    To address this challenge, we introduce a novel stochastic simulation approach based on a Feynman-Kac-like representation of the population dynamics. 
    This method, named the Feynman-Kac-inspired Gillespie's Stochastic Simulation Algorithm (FKG-SSA), always employs a fixed number of independently simulated cells for Monte Carlo computation of the system, resulting in a constant computational complexity regardless of the population size. 
    Furthermore, we theoretically show the statistical consistency of the proposed method, indicating its accuracy and reliability.
    Finally, a couple of biologically relevant numerical examples are presented to illustrate the approach. 
    Overall, the proposed FKG-SSA effectively addresses the challenge of simulating growing cell populations, providing a solid foundation for better analysis of these systems.
\end{abstract}



\section{Introduction}

Mathematical modeling and analysis have long been recognized as powerful tools for investigating living systems, providing insights into key biological mechanisms across various domains, such as oncology \cite{chen2024frequency, lakatos2020evolutionary, luksza2017neoantigen}, biological rhythms \cite{ goldbeter2002computational, song2024s, novak2008design}, and epidemiology \cite{ross1916application,kermack1927contribution}.
In studying cell populations, two major types of processes require detailed analysis: intracellular chemical reactions and cellular events (e.g., cell division, mutation, and death).
Over the past half-century, the dynamics of intracellular chemical reactions have been extensively investigated through theoretical analysis \cite{feinberg1987chemical,horn1972general, li2016systematic1, ke2019complex, fang2019lyapunov} and numerical simulations \cite{gillespie1977exact,gillespie2001approximate,rathinam2003stiffness,haseltine2002approximate, fang2023convergence,zechner2014uncoupled, fang2024advanced}.
Furthermore, rational engineering strategies for living cells have also been successfully proposed from this perspective \cite{ma2009defining,aoki2019universal,gardner2000construction,elowitz2000synthetic}.
In contrast, the research on exploring the combining effects of these two key types of processes has only recently gained momentum \cite{anderson2023stochastic,duso2020stochastic,lunz2021beyond,aditya2022using,piho2024feedback, jia2023coupling, gupta2022frequency}. 
Many studies have shown that cellular events play a significant role in regulating the cell population dynamics and that they can fundamentally change the system behaviors depending on their presence \cite{lunz2021beyond,piho2024feedback,jia2023coupling}. 
These facts underscore the need for more systematic and in-depth studies of cell populations that consider both major process types.

Cell population systems are often complicated due to the inherent non-linearity and randomness in intracellular chemical reactions and cellular events \cite{mcadams1997stochastic,arkin1998stochastic,fedoroff2002small}. 
For such complicated systems, numerical simulation is a promising approach for gaining insights through quantitative analysis. 
Following this line, many numerical approaches have been proposed for dividing/growing cell populations.
One popular approach is the agent-based simulation algorithm (see e.g., \cite{gorochowski2012bsim,izaguirre2004compucell,matyjaszkiewicz2017bsim,jafari2021multiscale,li2024computational,Piho2024agent} and also a recent review \cite{pleyer2023agent}), which exhaustively tracks intracellular chemical reactions and cellular events of every individual cell. 
This agent-based approach has been successfully applied to analyze many types of cell population systems (see the aforementioned
references).
However, it also suffers high computational costs due to the exhaustive tracking of individual cells, making them impractical for fast or even exponentially growing cell populations. 
Finite-state projection (FSP) \cite{munsky2006finite,piho2024feedback,ruess2023stochastic,lunz2021beyond}
is another popular approach for computing population systems, which directly utilizes numerical solvers to the differential equations characterizing the expected population dynamics.
This approach is generally accurate and reliable when the differential equation is truncated to a reasonably large size for computation.
However, its computational complexity often scales exponentially with the number of chemical species, rendering it intractable for complicated systems consisting of many chemical species.
To mitigate computational cost, one can also use the moment-closure approach \cite{duso2020stochastic,bronstein2018variational}, which closes and tracks the moment dynamics of the aforementioned differential equation.
While computationally efficient, moment closure lacks theoretical guarantees for its performance, and, therefore, its result is not always reliable. 
Overall, there is still a lack of efficient and reliable computation approaches for complicated cell population systems consisting of intracellular chemical reactions and cellular events.

To address this computational challenge, we propose a novel fixed-budget simulation algorithm for growing cell populations, which avoids tracking all the individual cells while still providing accurate and reliable numerical results. 
The method is motivated by the observation that the exhaustive tracking of all the cells is not necessary when only intracellular reactions are present. 
In such cases, the mean population dynamics is characterized by the master equation of the stochastic intracellular reaction processes.
Therefore, it can be effectively solved by the well-known Gillespie's stochastic simulation algorithm (SSA) \cite{gillespie1977exact}, which approximates the solution using a fixed number of simulated cells. 
The number of simulated cells can be significantly less than the actual cell population, ensuring computational efficiency. 
Moreover, its accuracy is also guaranteed by the law of large numbers.

Inspired by this observation, we first derived a Feynman-Kac-like formula for the dynamics of the mean cell populations consisting of both intracellular reactions and cellular events. 
The Feynman-Kac-like forintermula provides a probabilistic interpretation of the solution of the mean population dynamics through a modified stochastic cell system.
This enables us to further derive a stochastic simulation algorithm for the system that only requires simulating a fixed number of the modified cell systems.
We name this method the Feynman-Kac-inspired Gillespie's Stochastic Simulation Algorithm (FKG-SSA).
Again, the number of simulated cells in the FKG-SSA can be significantly less than the actual cell population size, allowing for a much reduced computational cost.
Furthermore, we show the convergence of the FKG-SSA to the exact solution through a series of rigorous theoretical analyses, which demonstrates the reliability of the approach. 
Notably, the theoretical results are not trivial consequences of the law of large numbers due to the interactions among the simulated samples. 
Moreover, the interacting term is not globally Lipschitz, adding more challenges to the analysis. 
In this paper, we prove the non-trivial convergence result by adopting some sophisticated techniques from the mean-field system analysis \cite{buckdahn2009mean, antonelli2002rate,bossy1997stochastic}.
The efficiency and accuracy of our approach are also illustrated by several biologically relevant numerical examples, which show that our method can be orders of magnitude more efficient than the agent-based simulation method while maintaining the same accuracy level. 
Overall, our FKG-SSA effectively addresses the computational challenge associated with growing cell populations.

The remainder of this paper is organized as follows. 
\Cref{section modeling} first briefly reviews the modeling of stochastic cell population systems and their mean population dynamics, following the literature \cite{duso2020stochastic,ruess2023stochastic}.
Then, in \Cref{section the algorithm}, we proposed a fixed-budget simulation algorithm (i.e., the FKG-SSA) for the growing cell populations from the perspective of the Feynman-Kac formula.
A couple of numerical examples are presented in \Cref{section numerical examples} to illustrate the accuracy and efficiency of our algorithm.
Finally, \Cref{section conclusion} concludes the paper.
For better readability, we provide most mathematical proofs in the appendix.

Some frequently used notations are summarized as follows. 
$\left(\Omega, \mathcal F, \{\mathcal F_t\}_{t\geq 0}, \mathbb P\right)$ is the filtered probability space, where $\Omega$ is the sample space, $\mathcal F$ is the sigma algebra on $\Omega$, $\{\mathcal F_t\}_{t\geq 0}$ is the filtration, and $\mathbb P$ is the probability measure. 
$\mathbbold{1}_{a}(x)$ is the  indicator function, which equals $1$ if $x=a$, and $0$ otherwise.
``$\wedge$" denotes the minimum of two quantities (i.e., $a\wedge b = \min(a,b)$).
$L^1\left(\mathbb Z^d_{\geq 0}\right)$ is the space of the functions from $\mathbb Z^d_{\geq 0}$ to $\mathbb R$ with finite $L^1$-norm.
The inner product $<\cdot, \cdot>$ for two functions is defined by $<f, g> \triangleq  \sum_{x\in\mathbb Z^{d}_{\geq 0}} f(x)g(x)$.
For a bounded operator $A$, we define the operator exponential by $e^{A}\triangleq \sum_{n=0}^{\infty} \frac{A^{n}}{n!}$, which is again a bounded operator.
The notations $\|\cdot\|_1$, $\|\cdot\|_2$, and $\|\cdot\|_\infty$ indicate the $L^1$-norm, $L^2$-norm, and the $L^{\infty}$-norm, respectively.

\section{Growing cell population modeling via chemical reaction network theory}\label{section modeling}

\subsection{Stochastic modeling of cell populations}

Here, we introduce the cell population modeling presented in \cite{duso2020stochastic,ruess2023stochastic}.
We consider the population systems in which individual cells involve stochastic chemical reactions, cell division, and death events.  
Also, we consider cell influx/migration from the environment to the considered system, which brings in new cells. 
In such a cell population, individual cells can have very different internal states.
Following \cite{duso2020stochastic,ruess2023stochastic}, we denote the internal state of the $i$-th cell as $\textbf{x}_{i}(t) = (\textbf{x}_{1,i}(t),\dots, \textbf{x}_{d, i}(t)) \in\mathbb Z_{\geq 0}^{d}$ with $d$ the number of the considered chemical species and $\textbf{x}_{k, i}(t)$ the number of $k$-th species in this cell at the time point $t$.
At the population level, we term $Y_t(x)$ (with $x\in\mathbb Z_{\geq 0}^d$) as the number of cells having the internal state $x$ at time $t$, i.e., $Y_t(x) = \# \{\textbf{x}_{i}(t)\big| \textbf{x}_{i}(t) = x\}$ with $\#$ indicating the cardinality of the set.
Thus, the function-valued variable $Y_t$ (from $\mathbb Z_{\geq 0}^d$ to $\mathbb Z_{\geq 0}$) represents the state of the cell population system at time $t$.
With these system-state notations, the system dynamics is provided in the following paragraphs; a summary is presented in \Cref{table:1}.

For intracellular chemical reaction processes, we consider that each cell has $r$ reactions:
\begin{align*}
	\nu_{1,j} S_1 + \dots + \nu_{d,j}S_d \to \nu'_{1,j} S_1 + \dots + \nu'_{d,j}S_d && j=1,\dots,r,
\end{align*}
where $S_1,\dots,S_d$ are $d$ different chemical species, and $\nu_{i,j}$ and $\nu'_{i,j}$ are the stoichiometric coefficients representing the numbers of molecules consumed and produced in the associated reaction.
Due to the low molecular counts, the firings of the reactions are usually modeled by Markovian jumps.
Specifically, each individual cell can change its state from $\textbf{x}_{i}(t)$ to $\textbf{x}_{i}(t) + \zeta_j$ at a rate of $\lambda^\text{react}_{j}(\textbf{x}_{i}(t)) $ (for $j=1,\dots, r$), where $\zeta_j \triangleq \nu'_{\cdot, j}- \nu_{\cdot, j}$ and $\lambda^\text{react}_{j}(x) $ is a non-negative function.
Then, at the population level, these chemical reactions cause the system state to change from $Y_t$ to $Y_t + \mathbbold 1_{x+\zeta_j} - \mathbbold 1_{x}$ at the rate $Y_t(x) \lambda^\text{react}_{j}(x)$ for any $x\in\mathbb Z^{d}_{\geq 0}$ and $j\in\{1,\dots,r\}$.
Here, $\mathbbold 1_{x}$ is the indicator function, which equals $1$ if its argument is $x$, and $0$ otherwise.
Such transitions are listed in the first row of \Cref{table:1}.

Then, we consider the cell death, division, and influx.
For the $i$-th cell, we consider that its death/remove rate is $\lambda^\text{death}(\textbf{x}_{i}(t))$ which depends on its internal state $\textbf{x}_{i}(t)$.
Then, at the population level, cell death can change the system state from $Y_t$ to $Y_t - \mathbbold 1_{x}$ at the rate of $Y_t(x)\lambda^\text{death}(x)$ for any $x\in \mathbb Z_{\geq 0}^d$.
For the division process, we consider that each cell has the division rate $\bar \lambda^\text{div}(\textbf{x}_{i}(t))$, and, after the division, the two daughter cells are assigned, in a specific order, the internal states $x'$ and $x''$ with probability $p^\text{div}(x', x''| \textbf{x}_{i}(t))$.
Thus, at the population level, cell division can alter the system state from $Y_t$ to $Y_{t} -  \mathbbold 1_{x} +  \mathbbold 1_{x'} +  \mathbbold 1_{x''}$ at the rate $Y_t(x) \left[\lambda^\text{div}(x, x', x'')+\lambda^\text{div}(x, x'', x')\right]$ (with $\lambda^\text{div}(x, x', x'')\triangleq  \bar \lambda^\text{div}(x) p^\text{div}(x', x'', | x)$)
for any $x$, $x'$, and $x''$ in $\mathbb Z^d_{\geq 0}$.
Finally, we consider that the system can add one cell in state $x$ (i.e., the state switches from $Y_t$ to $Y_t + \mathbbold 1_{x}$) at the rate of $\lambda^{in}(x)$ due to the cell migration from the environment to the considered cell population. 
The transitions caused by these cellular events are listed in the last three rows of \Cref{table:1}.

\begin{table}
\centering
\caption{Summary of all the transitions in the cell population system. Here, $y$ is any function from $\mathbb Z_{\geq 0}^d$ to $\mathbb Z_{\geq 0}$, the notation $\mathbbold 1_{a}$ represents the indicator function, and $x$, $x'$, and $x''$ are the states in $\mathbb Z^{d}_{\geq 0}$.}\label{tab1}%
\begin{tabular}{@{}lll@{}}
\toprule
Event &  Change of $Y_t$ & Rate\\
\midrule
$j$-th chemical reaction   &  $y \rightarrow y + \mathbbold 1_{x+\zeta_{j}} - \mathbbold 1_{x}$   & $y(x)\lambda^\text{react}_{j}(x)$   \\
cell death    & $y \rightarrow y - \mathbbold 1_{x}$  & $y(x)\lambda^\text{death}(x)$\\ 
cell division \& mutation  &  $y \rightarrow y - \mathbbold 1_{x} + \mathbbold 1_{x'} + \mathbbold 1_{x''}$ & $y(x) \left[\lambda^\text{div}(x, x', x'')+\lambda^\text{div}(x, x'', x')\right]$ \\
cell influx  &  $y \rightarrow y + \mathbbold 1_{x}$ & $\lambda^\text{in}(x)$ \\
\bottomrule
\end{tabular}
\label{table:1}
\end{table}

Given the discussion above, we can provide a dynamical equation for $Y_t$ as follows:
\begin{align}\label{eq. stochastic system}
    Y_t =& Y_0 + \sum_{x\in\mathbb Z^{d}_{\geq 0}} \sum_{j=1}^r \left(\mathbbold 1_{x+\zeta_j} - \mathbbold 1_{x} \right) R^\text{react}_{x,j}\left(\int_{0}^t Y_s(x)\lambda_j^\text{react}(x) \dd s \right) \\
    &- \sum_{x\in\mathbb Z^{d}_{\geq 0}}
    \mathbbold 1_{x} R^\text{death}_{x}\left(\int_{0}^t Y_s(x)\lambda^\text{death}(x) \dd s \right)\notag
    +  \sum_{x\in\mathbb Z^{d}_{\geq 0}}
    \mathbbold 1_{x} R^\text{in}_{x}\left(\int_{0}^t \lambda^\text{in}(x) \dd s \right)\notag
    \\
    &+ \sum_{x\in\mathbb Z^{d}_{\geq 0}}
    \sum_{x'\in\mathbb Z^{d}_{\geq 0}}
    \sum_{x''\in\mathbb Z^{d}_{\geq 0}}
    \left(\mathbbold 1_{x'} + \mathbbold 1_{x''} - \mathbbold 1_{x} \right) 
    R^\text{div}_{x,x',x''}\left(\int_{0}^t Y_s(x)\lambda^\text{div}(x, x', x'') \dd s \right) \notag
\end{align}
where $Y_0$ is the initial condition, and $\{R^\text{react}_{x,j}(t), R^\text{death}_{x}(t), R^\text{div}_{x,x',x''}(t), R^\text{in}_{x}(t) \}_{x,x',x'',j}$
are independent unit rate Poisson processes. 
In this equation, the last term in the first line represents chemical reaction processes, the second line represents cell death and influx, and the last line represents cell division events. 
In general, $Y_t$ is a Continuous-Time Markov Chain (CTMC) with the state in the functional space $L^1\left(\mathbb Z_{\geq 0}^d\right)$, where $L^1\left(\mathbb Z^d_{\geq 0}\right)$ is the set of the functions from $\mathbb Z^d_{\geq 0}$ to $\mathbb R$ with finite $L^1$-norm.
Due to the discrete and non-negativity nature of cell numbers, $Y_t(x)$ should always be a non-negative integer. 
Finding conditions for the non-explosivity of this stochastic process is challenging. 
Some results have been reported in \cite{anderson2023stochastic}.
For simplicity, this paper assumes that all the rate functions $\lambda^\text{react}_j(x)$, $\lambda^\text{death}(x)$, 
$\bar \lambda^\text{div}(x)$
and
$\lambda^{in}(x)$ are bounded in either the $L^1$- or $L^2$-norm, along with some other mild conditions (see \Cref{Assumption 1}). 
Under these conditions, $Y_t$ 
is almost surely non-explosive (see \Cref{proposition non-explosivity}). 

\begin{condition}\label{Assumption 1}
    $Y_0$ is a given non-negative integer-valued function in $L^1(\mathbb Z_{\geq 0}^d)$, and the rate functions satisfy
    \begin{itemize}
        \item $\lambda^\text{react}_{j}(x) = 0 $ when $x+\zeta_j \notin \mathbb Z^{d}_{\geq 0}$ or $x\notin \mathbb Z^{d}_{\geq 0}$.
        \item $\lambda^\text{react}_j(x)$, $\lambda^\text{death}(x)$, $\lambda^\text{div}(x,x',x'')$,
        $\bar \lambda^\text{div}(x)$ ($= \sum_{x'\in\mathbb Z_{\geq 0}^{d}} \sum_{x'' \in \mathbb Z_{\geq 0}^d}\lambda^\text{div}(x,x',x'')$)
        and $\lambda^{in}(x)$
        are non-negative and upper-bounded.
        \item $\lambda^{in}(x)$ has a finite $L^1$-norm.
    \end{itemize}
\end{condition}

\begin{proposition}\label{proposition non-explosivity}
    Under \Cref{Assumption 1}, the process $Y_t$ is almost surely non-explosive, and $\mathbb E\left[ \|Y_t\|_1\right] < +\infty$ for any $t\geq 0$. 
\end{proposition}
\begin{proof}
    The proof is in the appendix. 
\end{proof}

\subsection{Mean dynamics of cell population systems}

In many biological studies, scientists are often interested in investigating the mean dynamics of $Y_t$, particularly when the population size is large and the randomness can be accurately averaged out.
We denote the mean state by $n_t\triangleq  \mathbb E\left[ Y_t \right]$, or equivalently $n_t(x) \triangleq \mathbbold E\left[Y_t(x)\right]$ representing the expected number of cells in the state $x$ at time $t$.
Based on \eqref{eq. stochastic system}, the process $n_t$ satisfies the differential equation 
\begin{align}
    \frac{\dd }{\dd t} n_t(x)
    &= \sum_{j=1}^{r}\left[
    \lambda_{j}^{\text{react}}(x-\zeta_{j})n_t(x-\zeta_j)-\lambda_{j}^{\text{react}}(x)n_t(x)\right] 
    - \lambda^{\text{death}}(x)n_t(x) \notag\\
    &\quad - n_t(x)\sum_{x'\in\mathbb Z_{\geq 0}^d}\sum_{x''\in\mathbb Z_{\geq 0}^d}\lambda^{\text{div}}(x, x',x'') 
    + \sum_{\tilde{x}\in\mathbb Z_{\geq 0}^d}\sum_{\tilde{x}''\in\mathbb Z_{\geq 0}^d}\lambda^{\text{div}}(\tilde{x},x,\tilde{x}'')n_t(\tilde x) \notag \\
    &\quad + \sum_{\tilde{x}\in\mathbb Z_{\geq 0}^d}\sum_{\tilde{x}'\in\mathbb Z_{\geq 0}^d}
    \lambda^\text{div}(\tilde{x},\tilde{x}',x)n_t(\tilde x)
    + \lambda^{\text{in}}(x)
    \qquad\qquad\qquad \quad
    \forall (t,x)\in[0,+\infty) \times \mathbb Z^d_{\geq 0}.
    \notag 
\end{align}
This system is essentially an infinite-dimensional ODE system. 
By defining a reaction-related operator $A_\text{react}$ and a division-relevant operator $A_\text{div}$ as
\begin{align*}
    &[A_\text{react} f] (x) = \sum_{j=1}^{r}
    \lambda_{j}^{\text{react}}(x-\zeta_{j})f(x-\zeta_j)-\lambda_{j}^{\text{react}}(x)f(x) 
    \\
    &[A_\text{div} f] (x) = - f(x)\sum_{x'}\sum_{x''}\lambda^{\text{div}}(x, x',x'')
    + \sum_{\tilde{x}} f(\tilde x) 
     \sum_{\tilde{x}''}\left[\lambda^{\text{div}}(\tilde{x},x,\tilde{x}'')
    + \lambda^{\text{div}}(\tilde{x},\tilde{x}'', x)
    \right]
\end{align*}
for all $f\in L^1\left(\mathbb Z^{d}_{\geq 0}\right) $, we can rewrite this dynamical equation for $n_t$ by 
\begin{align}\label{eq n_t dynamics}
    \frac{\dd }{\dd t} n_t
    = \left(A_\text{react} - \lambda^\text{death} + A_\text{div}\right) n_t + \lambda^\text{in}    \qquad \forall t\geq 0 
    && \text{and}&&
    \quad n_0= \mu
\end{align}
with $\mu$ the non-negative initial condition defined by $\mu(x) = \mathbb E\left[ Y_0(x) \right]$.
Under \Cref{Assumption 1}, this differential system has a unique solution (see \Cref{proposition uniqueness of the mean dynamics}).

\begin{proposition}\label{proposition uniqueness of the mean dynamics}
    Under \Cref{Assumption 1}, 
    the differential system \eqref{eq n_t dynamics} has a unique solution satisfying $n_t \in L^1\left(\mathbb Z^{d}_{\geq 0}\right)$ for all $t\geq 0$, and this solution is given by $\mathbb E\left[Y_t\right]$.
\end{proposition}

\begin{proof}
    The proof is in the appendix.
\end{proof}

Due to certain technical reasons (which will be detailed in \Cref{section algorithm for the whole system}), we impose some constraints on the support of the function $\lambda^\text{in}(x)$ and the initial condition $\mu(x)$.

\begin{condition}\label{assumption 2}
    $\lambda^\text{in}(x)$ has finite support, and this support is a subset of the support of ${\mu(x)}$.
\end{condition}

\subsection{Computational challenge for the mean population dynamics $n_t$}
\label{section current challenges}

Numerical simulation approaches for system \eqref{eq n_t dynamics} include 
the Finite State Projection (FSP) \cite{munsky2006finite, piho2024feedback,ruess2023stochastic,lunz2021beyond}, agent-based simulation approach \cite{gorochowski2012bsim,izaguirre2004compucell,jafari2021multiscale,li2024computational,Piho2024agent, pleyer2023agent}, and the moment closure method \cite{duso2020stochastic}.
As demonstrated in the introduction, these methods all suffer many challenging in simulating growing cell populations.
Specifically, the FSP requires solving \eqref{eq n_t dynamics} on a reasonably large truncated space of $\mathbb Z_{\geq 0}^{d}$, which can lead to huge computational complexity when many chemical species are involved (i.e., $d$ is large).
The moment closure is much more efficient, as it only tracks the first few moments of $n_t$; however, its results lack theoretical guarantees, raising concerns about its reliability.
The agent-based simulation approach \cite{piho2024feedback} is a sampling method that applies the stochastic simulation algorithm (SSA) \cite{gillespie1977exact} to the CTMC \eqref{eq. stochastic system} and approximates $n_t$ using the mean of the simulated samples. 
This approach requires tracking all the individual cells and, therefore, can also be computationally inefficient for rapidly growing cell populations. 
In particular, the computational complexity of the SSA scales with the product of the number of jumps over the whole time period and the number of possible jump directions \cite[Section III]{sanft2015constant}, both of which grow linearly with the size of cell populations in our problem. 
Consequently, the agent-based method has a quadratic complexity with respect to the population size, making it impractical for even medium large cell populations. 
Overall, there is still a lack of efficient and reliable computation approaches for \eqref{eq n_t dynamics}.

\section{Fixed-budget simulation algorithm via the Feynman-Kac representation}\label{section the algorithm}

This section is devoted to developing a reliable fixed-budget simulation algorithm for the population system \eqref{eq n_t dynamics}.
This task is not impossible, and
such algorithms already exist for systems involving only chemical reactions.
In this case, the cell population size is fixed, and the equation \eqref{eq n_t dynamics} becomes the well-known Chemical Master Equation (CME) 
$\frac{\dd }{\dd t} n_t =  A_\text{react} n_t$ with  $n_0 = \mu$ \cite{anderson2015stochastic},
whoese normalized solution characterizes the distribution of the stochastic process 
\begin{align}\label{eq. X(t)}
    X(t) = X(0) + \sum_{j=1}^r \zeta_j \tilde  R_{j}\left(\int_{0}^t \lambda^\text{react}_j(X(s)) \dd s \right)
    &&  X(0) \sim {\mu (x)}/{\|\mu\|_1}
\end{align}
with $\{\tilde R_{j}(t)\}_{j=1,\dots, r}$ independent unit rate Poisson process. 
In other words, the solution of the CME has the representation $n_t(x) = \|\mu\|_1\, \mathbb P\left( 
X(t) = x \right)$, suggesting that the CME can be numerically solved by generating simulated samples for $X(t)$.
In this approach, one can set a fixed size of the simulation samples for $X(t)$ regardless of the actual cell population size in the system, thereby resulting in a fixed-budget simulation algorithm.

When other cellular events exist, the representation  $n_t(x) = \|\mu\|_1\, \mathbb P\left( 
X(t) = x \right)$ no longer holds.
In the following, we employ the Feynman-Kac formula to provide similar probabilistic representations of $n_t$ and provide fixed-budget simulation algorithms for the population dynamics \eqref{eq n_t dynamics}.
In \Cref{section reactions + death}---\ref{section algorithm for the whole system}, we gradually take the other cellular events (cell death, division, and influx) into consideration to develop the algorithm. 
In \Cref{section algorithm + restarting}, we introduce a resampling/restarting strategy to the algorithm for improved accuracy.

\subsection{Algorithm for systems with only chemical reactions and cell death}
\label{section reactions + death}
When the population system only evolves chemical reactions and cell death, equation \eqref{eq n_t dynamics} becomes 
\begin{align}\label{eq. n_t with reaction + death}
    \frac{\dd }{\dd t} n_t
    = \left(A_\text{react} - \lambda^\text{death} \right) n_t
    && \text{with } n_0 = \mu. 
\end{align}
For this equation, a probabilistic representation of $n_t$ was ``magically" presented in \cite[Lemma 1]{rathinam2021state} (also see \eqref{eq. solution of n_t reaction + death} later in this section), when investigating filtering problems. 
To gain a better understanding of this representation and inspire the derivation of a similar formula for the general dynamics \eqref{eq n_t dynamics}, we revisit it from the  Feynman-Kac perspective \cite{kac1949distributions,karandikar1987feynman}. 

First, we consider the adjoint equation of \eqref{eq. n_t with reaction + death} expressed as
\begin{align}\label{eq. adjoint equation reaction + death}
    \left\{
    \begin{array}{l}
        \dot 
        \phi_t
        =  -\left(A^*_\text{react} - \lambda^\text{death} \right) \phi_t \\
          \phi_T(x) = g(x)
    \end{array}
    \right.
    && \text{for all } t\in [0, T] \text{~and~} x\in\mathbb Z^d_{\geq 0}
\end{align}
where $T$ is a given terminal time, $g(x)$ is a given bounded function on $\mathbb Z^d_{\geq 0}$, and 
$A^*_\text{react}$ is the adjoint operator of $A_\text{react}$ given by 
\begin{align*}
    \left[A^*_\text{react} f\right] (x)
    = \sum_{j=1}^{r}\left[
    f(x+v_{k})-f(x)\right] 
    \lambda_{j}^{\text{react}}(x)
    && \text{for any bounded function $f$ on $\mathbb Z^d_{\geq 0}$.}
\end{align*}
Notice that $A^*_\text{react}$ is the generator of the stochastic process $X(t)$.
Then, based on the Feynman-Kac formula, this adjoint equation has a unique solution given in the following proposition. 

\begin{proposition}\label{proposition solution of the adjoint equation reaction + death}
    Under \Cref{Assumption 1}, for any given $T$ and bounded function $g$, Eq. \eqref{eq. adjoint equation reaction + death} has a unique bounded solution with the expression given by 
    \begin{align}\label{eq solution of the adjoint equation reaction + death}
        \phi_t(x) = \mathbb E\left[g(X(T))\exp\left(\int_{t}^{T}-\lambda^{\text{death}}(X(s))ds\right) \bigg| X(t)=x\right],
        &&
        \forall (t,x)\in [0, T] \times \mathbb Z_{\geq 0}^d.
    \end{align}
\end{proposition}
\begin{proof}
    The solution is given by the Feynman-Kac formula (see a version in \cite[Theorem 4]{karandikar1987feynman}).
    By the boundedness of all the rate functions (\Cref{Assumption 1}), the operator $\left(A^*_\text{react} - \lambda^\text{death} \right)$ is bounded, which implies the uniqueness of the solution. 
\end{proof}

Then, the solution of the original equation \eqref{eq. n_t with reaction + death} can be constructed based on this adjoint equation. 
In  \Cref{proposition adjoint property reaction + death}, we provide several conditions related to the adjoint equation for checking whether an $L^1 \left(\mathbb Z^d_{\geq 0}\right)$-valued process $u_t$ solves the population dynamics \eqref{eq. n_t with reaction + death}.
Further, by the Feynman-Kac representation, the only process satisfying these conditions (i.e., the unique solution of \eqref{eq. n_t with reaction + death}) is given by
\begin{align}\label{eq. solution of n_t reaction + death}
    n_t(x) ~=~ \|\mu\|_1\, \mathbb E\left[
        \mathbbold 1_{x} \left( X(t) \right)\exp\left(\int_{0}^{t}-\lambda^{\text{death}}(X(s))ds\right)
    \right]
\end{align}
for all $t\in[0, +\infty)$ and $x\in\mathbb Z^{d}_{\geq 0}$ (see \Cref{them solution of n_t reaction + death}).

\begin{proposition}\label{proposition adjoint property reaction + death}
    Denote $\phi^{T,g}_t$ as the bounded solution of \eqref{eq. adjoint equation reaction + death} with the terminal time $T$ and terminal condition $g$.
    Then  an $L^1 \left(\mathbb Z^d_{\geq 0}\right)$-valued time-differentiable process $u_t$ solves \eqref{eq. n_t with reaction + death}, if $u_0=\mu$ and
    $\frac{d}{d t}<\phi^{T,g}_t, u_t>=0$ for any $T>0$, $t\in[0,T]$, and bounded function $g(x)$.
\end{proposition}
\begin{proof}
    The proof is presented in the appendix.
\end{proof}

\begin{theorem}\label{them solution of n_t reaction + death}
    The process $n_t$ given by \eqref{eq. solution of n_t reaction + death} is the only process satisfying the conditions in \Cref{proposition adjoint property reaction + death}, and, therefore, it is the only $L^1\left(\mathbb Z^{d}_{\geq 0}\right)$-valued process solving \eqref{eq. n_t with reaction + death}. 
\end{theorem}
\begin{proof}
    The proof is presented in the appendix. 
\end{proof}

The probabilistic representation \eqref{eq. solution of n_t reaction + death} allows for a fixed-budget simulation algorithm for the population dynamics \eqref{eq. n_t with reaction + death} (see \Cref{algorithm only death}), where only a fixed size of simulation samples for $X(t)$ are generated and used for constructing the solution. 
The algorithm has a complexity of $\mathcal{O}(N)$, where $N$ represents the number of simulation samples for $X(t)$, independent of the actual cell population size (i.e., $|n_t|_1$). 
This indicates that the method is significantly more efficient than agent-based simulation methods when dealing with large cell populations.

\begin{algorithm2e}
\SetAlgoLined
Input the initial conditions of $\mu(x)$,  sample size $N$, and final time $T$. \\
Simulate N trajectories of $X(t)$ up to time $T$ with initial states sampled from the distribution $\mu(x)/\|\mu\|_1$. Denote them by $\texttt{x}_i(t)$ ($i=1,\dots, N$). \\
Calculate weights $w_{i}(t) = \exp\left(\int_{0}^{t}-\lambda^{\text{death}}(\texttt{x}_{i}(s))ds\right)$. \\
Output the solution: $\hat n_t(x) = 
 \frac{\|\mu\|_1}{N}\sum_{i=1}^{N} \mathbbold 1_{x}(\texttt{x}_{i}(t))w_{i}(t)$.
\caption{Fixed-budget simulation algorithm for \eqref{eq. n_t with reaction + death}}
\label{algorithm only death}
\end{algorithm2e}

The above derivations (\Cref{proposition solution of the adjoint equation reaction + death}, \Cref{proposition adjoint property reaction + death}, and \Cref{them solution of n_t reaction + death}) demonstrate that the probabilistic representation \eqref{eq. solution of n_t reaction + death} is a direct consequence of the Feynman-Kac representation for the adjoint equation \eqref{eq. adjoint equation reaction + death}.
From this perspective, we can interpret \eqref{eq. solution of n_t reaction + death} as a Feynman-Kac-like representation for Eq. \eqref{eq. n_t with reaction + death}, and we therefore refer to \Cref{algorithm only death} as  the Feynman-Kac-inspired Gillespie's Stochastic Simulation Algorithm (FKG-SSA).

\begin{remark}
    We can also interpret \eqref{eq. solution of n_t reaction + death} from the biological perspective.
    When only chemical reactions and cell death exist, $X(t)$ describes the internal state of a living cell up to its death.
    After the cell death, we can view the cell as becoming an imaginal ``ghost" cell, whose internal state continues evolving according to the dynamics of $X(t)$.
    In this regard, the term $\exp\left(\int_{0}^{t}-\lambda^{\text{death}}(X(s))ds\right)$ represents the likelihood of a cell surviving up to time $t$, given the entire trajectory of its internal state, both as a normal cell and as a ``ghost" cell.
    Consequently, the term $\|\mu\|_1E\left[
        \mathbbold 1_{x} \left( X(t) \right)\exp\left(\int_{0}^{t}-\lambda^{\text{death}}(X(s))ds\right)
    \right]$ (on the right-hand side of \eqref{eq. solution of n_t reaction + death})
    represents the expected number of cells being in state $x$ and still alive, which corresponds to $n_t(x)$. 
\end{remark}

\subsection{Algorithm for systems with chemical reactions, cell death, and division}\label{section reaction + death + division}

When only the cell influx is absent, the population dynamics becomes
\begin{align}\label{eq. n_t with reaction, death, division}
    \frac{\dd }{\dd t} n_t
    = \left(A_\text{react} - \lambda^\text{death} + A_\text{div}\right) n_t
    && \text{with } n_0 = \mu. 
\end{align}
To the best of our knowledge, a Feynman-Kac-like representation for this equation has not been proposed before.
Here, we derive such a representation and provide a fixed-budget simulation algorithm for this population system.

Notice that the adjoint operator of $\left(A_\text{react} + A_\text{div}\right)$ does not correspond to any stochastic process, which invalidates the straightforward application of the strategy in \Cref{section reactions + death}.
Thus, we first reformulate \eqref{eq. n_t with reaction, death, division} so that that strategy directly applies here. 
We denote a new operator $\tilde A \triangleq A_\text{react}  + A_\text{div} - \bar \lambda^\text{div}$ (recall that $\bar \lambda^\text{div}(x) = \sum_{x'\in\mathbb Z_{\geq 0}^{d}} \sum_{x'' \in \mathbb Z_{\geq 0}^d}\lambda^\text{div}(x,x',x'')$).
Then, the differential equation \eqref{eq. n_t with reaction, death, division} can be rewritten as 
$\frac{\dd }{\dd t} n_t
     = \left(\tilde A - \lambda^\text{death} +  \bar \lambda^\text{div} \right) n_t$
with $n_0 = \mu$.
We further denote $\tilde A^{*}$ as the adjoint operator of $\tilde A$, which has the expression
\begin{align*}
    \left[\tilde A^* f\right] (x)
    =& \sum_{j=1}^{r}\left[
    f(x+v_{k})-f(x)\right] 
    \lambda_{j}^{\text{react}}(x)+ \sum_{ x'\in\mathbb Z_{\geq 0}^d} \Big[f( x') - f(x)\Big]
    \lambda^{\overline{\text{{div}}}}(x, x')
\end{align*}
for any bounded function $f$ on $\mathbb Z^d_{\geq 0}$, where $\lambda^{\overline{\text{{div}}}}(x, x')\triangleq \sum_{ {x}''\in\mathbb Z_{\geq 0}^d}\lambda^{\text{div}}(x, x',\tilde{x}'') + \lambda^{\text{div}}(x, {x}'', {x}')$.
Apparently, this adjoint operator $\tilde A^{*}$ is the generator of the stochastic process
\begin{align}\label{eq tilde X}
    \tilde X(t) =& \tilde X(0) + \sum_{j=1}^r\zeta_j \tilde R_{j}\left(\int_{0}^{t}\lambda_{j}^{\text{react}}\left(\tilde X(s)\right)ds\right)  \\
    &+
    \sum_{x\in\mathbb Z_{\geq 0}^d}\sum_{x'\in\mathbb Z_{\geq 0}^d}(x'-x)\tilde R_{x,x'}\left(\int_{0}^t 
    \mathbbold 1_{x}\left(\tilde X(s)\right) 
    \lambda^{\overline{\text{{div}}}}(x, x')
    \dd s \right) \notag 
\end{align}
where $\tilde R_{j}$ (for $j=1,\dots, r$) and  $ \tilde R_{x,x'}$ (for $x, x' \in \mathbb Z_{\geq 0}^d$) are independent unit rate Poisson processes, and $\tilde X(0)$ follows the distribution $\mu(x)/\|\mu\|_1$.
Under \Cref{Assumption 1}, the total transition rate of $\tilde X(t)$ is always upper bounded, and, therefore, it is almost surely non-explosive.

With these new notations, the adjoin equation of \eqref{eq. n_t with reaction, death, division} is given by 
\begin{align}\label{eq. adjoint equation reaction, death, division}
    \left\{
    \begin{array}{l}
        \dot 
        \phi_t
        =  -\left(\tilde A^* - \lambda^\text{death} + \bar \lambda^\text{div} \right) \phi_t \\
          \phi_T(x) = g(x)
    \end{array}
    \right.
    && \text{for all } t\in [0, T] \text{~and~} x\in\mathbb Z^d_{\geq 0}
\end{align}
where $T$ is a given terminal time, $g(x)$ is a given bounded function on $\mathbb Z^d_{\geq 0}$.
Based on the Feynman-Kac formula, this adjoint equation also has a unique solution (see \Cref{proposition solution of the adjoint equation reaction + death + division}).

\begin{proposition}\label{proposition solution of the adjoint equation reaction + death + division}
    Under \Cref{Assumption 1}, for any given $T$ and bounded function $g$, Eq. \eqref{eq. adjoint equation reaction, death, division} has a unique bounded solution with the expression given by 
    \begin{align}\label{eq solution of the adjoint equation reaction, death, division}
        \phi_t(x) = \mathbb E\left[g\left(\tilde X(T)\right)\exp\left(\int_{t}^{T} \bar \lambda^\text{div}\left(\tilde X(s)\right)-\lambda^{\text{death}}\left(\tilde X(s)\right) \dd s \right) \bigg| \tilde X(t)=x\right]
    \end{align}
    for all $t\in [0, T]$ and $x\in\mathbb Z_{\geq 0}^d$.
\end{proposition}
\begin{proof}
    The solution is given by the Feynman-Kac formula (see  \cite[Theorem 4]{karandikar1987feynman}).
    The uniqueness follows immediately from the boundedness of $\tilde A^*$ (suggested by \Cref{Assumption 1}).
\end{proof}

Following the strategy in \Cref{section reactions + death}, we can provide adjoint-equation-relevant conditions for checking whether an $L^1 \left(\mathbb Z^d_{\geq 0}\right)$-valued process $u_t$ solves Eq. \eqref{eq. n_t with reaction, death, division} (see \Cref{proposition adjoint property reaction + death + division}).
Moreover, we can further show that the only process satisfying these conditions
is given by 
\begin{align}\label{eq. solution of n_t reaction, death, division}
    n_t(x) ~=~ 
    \|\mu\|_1\,
    \mathbb E\left[
        \mathbbold 1_{x} \left( \tilde X(t) \right)\exp\left(\int_{0}^{t}
        \bar \lambda^\text{div}\left(\tilde X(s)\right) -\lambda^{\text{death}}\left(\tilde X(s)\right) \dd s\right)
    \right]
\end{align}
for all $t\in[0, +\infty)$ and $x\in\mathbb Z^{d}_{\geq 0}$, which therefore uniquely solves \eqref{eq. n_t with reaction, death, division} (see \Cref{them solution of n_t reaction + death + division}).

\begin{proposition}\label{proposition adjoint property reaction + death + division}
    Denote $\phi^{T,g}_t$ as the bounded solution of \eqref{eq. adjoint equation reaction, death, division} with the terminal time $T$ and terminal condition $g$.
    Then, an $L^1 \left(\mathbb Z^d_{\geq 0}\right)$-valued time-differentiable process $u_t$ solves \eqref{eq. solution of n_t reaction, death, division}, if $u_0=\mu$ and
    $\frac{d}{d t}<\phi^{T,g}_t, u_t>=0$ for any $T>0$, $t\in[0,T]$, and bounded function $g(x)$.
\end{proposition}
\begin{proof}
    The proof is almost the same as the one of \Cref{proposition adjoint property reaction + death}. 
    We leave it to readers.
\end{proof}

\begin{theorem}\label{them solution of n_t reaction + death + division}
    The process $n_t$ given by \eqref{eq solution of the adjoint equation reaction, death, division} satisfies the conditions in \Cref{proposition adjoint property reaction + death + division}, and, therefore, it is the only $L^1\left(\mathbb Z^{d}_{\geq 0}\right)$-valued process solving \eqref{eq. n_t with reaction, death, division}. 
\end{theorem}

\begin{proof}
    The proof is almost the same as the one of \Cref{them solution of n_t reaction + death}.
    We leave it to readers.
\end{proof}

Similar to \eqref{eq. n_t with reaction + death}, the Feynman-Kac-like representation \eqref{eq. n_t with reaction, death, division} enables a fixed-budget simulation algorithm for the population dynamics \eqref{eq. n_t with reaction, death, division} (see \Cref{algorithm death and division}).
As in the previous subsection, we refer to this algorithm as the Feynman-Kac-inspired Gillespie's Stochastic Simulation Algorithm (FKG-SSA) for \eqref{eq. n_t with reaction, death, division}.
In detail, this algorithm generates a fixed number of simulation samples for $\tilde X(t)$ and computes their associated exponential terms $\exp\left(\int_{0}^{t}
        \bar \lambda^\text{div}\left(\tilde X(s)\right) -\lambda^{\text{death}}\left(\tilde X(s)\right) \dd s\right)$
which can be viewed as weights for these simulation samples. 
Then, \Cref{algorithm death and division} applies these weighted samples to construct the solution of \eqref{eq. n_t with reaction, death, division} based on \eqref{eq. solution of n_t reaction, death, division}.
This method again has a computational complexity of $\mathcal O(N)$ with $N$ the sample size, independent of the actual cell population (i.e., $\|n_t\|_1$).
Consequently, when applied to rapidly dividing large-scale populations, this method can be significantly more efficient than the agent-based simulation algorithms that track every individual cell.

\begin{algorithm2e}
\caption{FKG-SSA for \eqref{eq. n_t with reaction, death, division}}
\label{algorithm death and division}
\SetAlgoLined
Input the initial conditions of $\mu(x)$,  sample size $N$, and final time $T$. \;
Simulate N trajectories of $\tilde{X}(t)$ up to time T with initial states sampled from the distribution $\mu(x)/\|\mu\|_1$. Denote them by $\texttt{x}_i(t)$ ($i=1,\dots, N$). \\
Calculate weights $w_{i}(t) = \exp\left(\int_{0}^{t} \bar \lambda^\text{div}(\texttt{x}_{i}(s)) -\lambda^{\text{death}} (\texttt{x}_{i}(s))ds\right)$. \\
Output the solution: $ \hat n_{t}(x) =  \frac{\|\mu\|_1}{N}\sum_{i=1}^{N} \mathbbold 1_{x}(\texttt{x}_{i}(t))w_{i}(t)$.
\end{algorithm2e}


\begin{remark}\label{remark interpretation of the algorithm reaction + death + division}
    We can also interpret the representation \eqref{eq. solution of n_t reaction, death, division} from the biological perspective. 
    First, compared with $X(t)$ described in \eqref{eq. X(t)}, the process $\tilde X(t)$ contains additional terms in the second row of \eqref{eq tilde X}, representing cell division.
    This structure allows for $\tilde X(t)$ to track the internal state of an immortal cell lineage.
    Specifically, when the cell divides, the process $\tilde X(t)$ starts to track the internal state of one daughter cell. 
    Notably, this cell lineage differs from those modeled by \eqref{eq. stochastic system} (or equivalently \eqref{eq n_t dynamics});
    the cells described by $\tilde X(t)$ do not undergo death, and their division rate is twice that of the cells modeled in \eqref{eq. stochastic system}.
    These discrepancies between $\tilde X(t)$ and the target biological cells are compensated via the weight term in \eqref{eq. solution of n_t reaction, death, division}.
    Specifically, the term $\exp\left(\int_{0}^{t} -\lambda^{\text{death}}\left(\tilde X(s)\right) \dd s\right)$  represents the likelihood of this cell lineage surviving up to time $t$.
    The remaining part, $\exp\left(\int_{0}^{t} \bar \lambda^{\text{div}}\left(\tilde X(s)\right) \dd s\right)$,
    compensates for the inability to track both daughter cells and the mismatch in the division rates. 
    In particular, cells with higher division rates typically dominate the population \cite{piho2024feedback}. This is reflected in the algorithm, as it assigns greater rewards to cell lineages with faster division rates. 
\end{remark}

\subsection{Algorithm for the general cell population system \eqref{eq n_t dynamics}}
\label{section algorithm for the whole system}

This section focuses on the general population dynamics \eqref{eq n_t dynamics}.
Again, we first provide a Feynman-Kac-like representation for Eq. \eqref{eq n_t dynamics} and then utilize it to construct a fixed-budget simulation algorithm. 

We observe that \eqref{eq n_t dynamics} is a non-homogeneous linear ODE due to the influx term $\lambda^\text{in}(x)$, and its homogeneous version is presented and already solved in \Cref{section reaction + death + division}.
This suggests that Eq. \eqref{eq n_t dynamics} can be solved using the results provided in \Cref{section reaction + death + division}. 
Following this idea, we first write the solution of the non-homogeneous equation \eqref{eq n_t dynamics} under \Cref{Assumption 1} by 
\begin{align}\label{eq. general solution of the population dyanmics}
    n_t &= e^{\left(A_\text{react} - \lambda^\text{death} + A_\text{div} \right) t }\, \mu
    + \int_0^t e^{\left(A_\text{react} - \lambda^\text{death} + A_\text{div} \right) (t-s) } \lambda^\text{in} \, \dd s \notag \\
    &= e^{\left(A_\text{react} - \lambda^\text{death} + A_\text{div} \right) t } \,\mu
    + \int_0^t \sum_{z:\, \lambda^\text{in}(z) \neq 0 }  \lambda^\text{in}(z) e^{\left(A_\text{react} - \lambda^\text{death} + A_\text{div} \right) (t-s) }  \mathbbold 1_{z} \, \dd s
    && \forall t \geq 0,
\end{align}
where the first line follows immediately from the general solution to non-homogeneous ODEs, and the second line follows from the expression $\lambda^\text{in} = \sum_{z:\, \lambda^\text{in}(z) \neq 0 } \lambda^\text{in}(z) \mathbbold 1_{z}$.
We notice that $e^{\left(A_\text{react} - \lambda^\text{death} + A_\text{div} \right) t } \mu$ and $\lambda^\text{in} (z) e^{\left(A_\text{react} - \lambda^\text{death} + A_\text{div} \right) (t-s) } \mathbbold 1_{z}$ both solve the differential equation $\frac{\dd }{\dd t} u_t
    = \left(A_\text{react} - \lambda^\text{death} + A_\text{div}\right) u_t$ (i.e., Eq. \eqref{eq. n_t with reaction, death, division}) with the initial conditions being $u_0 = \mu$ and $u_s = \lambda^\text{in}(z) \mathbbold 1_{z}$, respectively.
Therefore, by \eqref{eq. solution of n_t reaction, death, division} and \Cref{them solution of n_t reaction + death + division}, we can provide Feynman-Kac-like representation for both quantities:
\begin{align*}
    &\left[e^{\left(A_\text{react} - \lambda^\text{death} + A_\text{div} \right) t } \mu\right] (x)
    ~=~ \|\mu\|_1\, \mathbb E\left[
        \mathbbold 1_{x} \left( \tilde X(t) \right)
        e^{\int_{0}^t \lambda^{D}\left(\tilde X(s)\right) \dd s}
        \right] 
\end{align*}
with $\lambda^{D}(x) \triangleq  \bar \lambda^\text{div}(x) - \lambda^\text{death}(x)$
and, for any $z$ where $\lambda^\text{in}(z) \neq 0$,
\begin{align*}
    \left[\lambda^\text{in}(z) e^{\left(A_\text{react} - \lambda^\text{death} + A_\text{div} \right) (t-s) } \mathbbold 1_{z}\right](x) 
    &= \lambda^\text{in}(z) \, \mathbb E\left[
        \mathbbold 1_{x} \left( \tilde X(t) \right)
        e^{\int_{s}^t \lambda^{D}\left(\tilde X(\tau)\right) \dd \tau}
        \Big| 
        \tilde X(s) = z
        \right] \\
    &= 
        \frac{\lambda^\text{in}(z) }{\mathbb P\left(\tilde X(s) = z\right)}
        \mathbb E\left[
        \mathbbold 1_{x} \left( \tilde X(t) \right)
        \mathbbold 1_{z} \left( \tilde X(s) \right)
        e^{\int_{s}^t \lambda^{D}\left(\tilde X(\tau)\right) \dd \tau}
        \right]
\end{align*}
where the last line follows from the fact that ${\mathbb P\left(\tilde X(s) = z\right)} \neq 0 $ for such state $z$ under \Cref{Assumption 1}--\ref{assumption 2} (see \Cref{proposition non-zero of the probability}).
Finally, by applying these representations to \eqref{eq. general solution of the population dyanmics}, we provide a Feynman-Kac-like representation for the solution of the general population dynamics \eqref{eq n_t dynamics}:
\begin{align}\label{eq. solution of n_t general case}
    &n_t(x) ~=~ 
    \|\mu\|_1\,
    \mathbb E\left[
        \mathbbold 1_{x} \left( \tilde X(t) \right)
        W_t
    \right] \\
     & W_t ~=~ e^{\int_{0}^t \lambda^{D}\left(\tilde X(s)\right) \dd s}
        + \int_{0}^{t} e^{\int_{s}^t \lambda^{D}\left(\tilde X(\tau)\right) \dd \tau}
        \frac{\lambda^\text{in}\left(\tilde X(s)\right)}{\|\mu\|_1\,\tilde p(t, \tilde X(s))}
        \dd s 
        && \text{and } W_0 =1
        \label{eq. weights of the final algorithm}
\end{align}
where $\lambda^{D}(x) \triangleq  \bar \lambda^\text{div}(x) - \lambda^\text{death}(x)$ and $\tilde p(t, z) \triangleq \mathbb P \left( \tilde X(t) = z\right).$
Notably, this weighting process $W_t$ can also be written as $\dot W_t = \left(\bar \lambda^\text{div}\left(\tilde X(t)\right) -\lambda^{\text{death}}\left(\tilde X(t)\right) \right)W_t + \frac{\lambda^\text{in}\left(\tilde X(t) \right)}{\|\mu\|_1 \, \tilde p\left(t, \tilde X(t) \right)}$.
The results of the above derivations are summarized in \Cref{them solution of n_t general case}.

\begin{proposition}\label{proposition non-zero of the probability}
    Under \Cref{Assumption 1}--\ref{assumption 2}, the stochastic process $\tilde X(t)$ given in \eqref{eq tilde X} satisfies
    $${\mathbb P\left(\tilde X(t) = z\right)} ~\geq~ e^{-t\left(\sum_{j=1}^r\|\lambda^\text{react}_j\|_\infty + 2\|\bar \lambda^\text{div}\|_\infty\right) } \mathbb P\left(\tilde X(0) = z\right) ~>~ 0$$ for all $t>0$ and $z$ in the support of $\mu(x)$ or the support of $\lambda^\text{in}(x)$.
\end{proposition}
\begin{proof}
    The proof is given in the appendix.
\end{proof}

\begin{theorem}\label{them solution of n_t general case}
    Under \Cref{Assumption 1}--\ref{assumption 2}, the process $n_t$ given in \eqref{eq. solution of n_t general case} is the unique $L^1\left(\mathbb Z^{d}_{\geq 0}\right)$-valued process solving \eqref{eq n_t dynamics}.
\end{theorem}
\begin{proof}
    The derivations for \eqref{eq. solution of n_t general case} have already shown that this $n_t$ solves \eqref{eq n_t dynamics}.
    Eq. \eqref{eq. general solution of the population dyanmics} together with the boundedness of $A_\text{react}$,   $\lambda^\text{death}$, and  $A_\text{div} $ (implied by \Cref{Assumption 1}) suggests that $n_t$ has a finite $L^1$-norm for any $t>0$.
    The uniqueness follows from \Cref{proposition uniqueness of the mean dynamics}.
\end{proof}

The representation \eqref{eq. solution of n_t general case} enables a fixed-budget simulation algorithm for the general population dynamics \eqref{eq n_t dynamics} (see \Cref{algorithm death and division and influx}), which is again referred to as the FKG-SSA.
This algorithm first generates a fixed number of simulation samples for $\tilde X(t)$ and then uses them to approximate the distribution $\tilde p(t,x)$ (see Line 3 in \Cref{algorithm death and division and influx}).
These samples and the approximated distribution are further applied to compute the weights according to \eqref{eq. weights of the final algorithm} (see Line 4 in \Cref{algorithm death and division and influx}).
Finally, an approximation to the solution of \eqref{eq n_t dynamics} is established based on the Feynman-Kac-like representation \eqref{eq. solution of n_t general case}.
This algorithm again has a computational complexity of $\mathcal O (N)$
with N the simulation sample size.
Consequently, when $N$ is fixed, our algorithm always has a fixed computational cost, irrelevant to the actual cell population size $\|n_t\|_1$.

\begin{algorithm2e}
\caption{FKG-SSA for the general population dynamics \eqref{eq n_t dynamics}}
\label{algorithm death and division and influx}
\SetAlgoLined
Input the initial condition $\mu(x)$,  sample size $N$, and final time $T$. \\
Simulate N trajectories of $\tilde{X}(t)$ up to time T with initial states sampled from the distribution $\mu(x)/\|\mu\|_1$. Denote them by $\texttt{x}_i(t)$ ($i=1,\dots, N$).  \\
Approximate the distribution of $\tilde X(t)$: $\hat p(t,x) = \frac{1}{N}\sum_{i=1}^{N} \mathbbold 1_{x}(\texttt{x}_{i}(t)).$ \\
Calculate weights $w_i(t)$ by solving the differential equation $\dot w_i(t) = \left(\bar \lambda^\text{div}(\texttt{x}_i(t)) - \lambda^\text{death}(\texttt{x}_i(t))\right) w_i(t) + \frac{\lambda^\text{in}(\texttt{x}_i(t))}{\|\mu\|_1 \hat p(t, \texttt{x}_i(t))}$ with $w_i(0)=1$.  \\
Output the solution: $ \hat n_{t}(x) =  \frac{\|\mu\|_1}{N}\sum_{i=1}^{N} \mathbbold 1_{x}(\texttt{x}_{i}(t))w_{i}(t)$.
\end{algorithm2e}

In addition to the computational advantage, the algorithm is also reliable, as demonstrated by its convergence to the exact solution of \eqref{eq n_t dynamics} at a rate of $\frac{1}{\sqrt{N}}$ (see \Cref{them convergence of the FKG-SSA}).
This convergence result is non-trivial, as the interactions among the samples complicate the analysis.
Specifically, the simulated samples interact through the approximation of $\tilde p(t,x)$ and the weight computation as described in Line 3--4 of \Cref{algorithm death and division and influx}.
Moreover, the interaction term $\frac{\lambda^\text{in}(x)}{\|\mu\|_1 \hat p(t, x)}$ in the weight dynamics is not globally Lipschitz, which adds more challenges. 
To analyze this interacting particle system, we adopt some techniques from the mean-field system analysis
\cite{buckdahn2009mean, antonelli2002rate,bossy1997stochastic} (see the proof of this theorem in the appendix).

\begin{theorem}\label{them convergence of the FKG-SSA}
    Assume \Cref{Assumption 1}---\ref{assumption 2} hold, and let $T$ be the final time in \Cref{algorithm death and division and influx}.
    Then, for any $t\in[0,T]$, there exist positive constants $C_1$, $C_2^z$ (for $z$ in the support of $\lambda^\text{in}$), and $C_3$ such that the output of \Cref{algorithm death and division and influx} (denoted by $\hat n_t$) satisfies
    \begin{align*}
        \mathbb E\left[
            \|\hat n_t - n_t\|^2_2
        \right] \leq 
         C_1 N^2
        \left(
        \sum_{z:\,\lambda^\text{in}(z)\neq 0}
         e^{ -C^z_{2} N}
        \right)
        +
        \frac{C_3}{N}  
        && \forall t\geq 0
    \end{align*}
    where $n_t$ is the solution of \eqref{eq n_t dynamics}, and $N$ is the sample size of \Cref{algorithm death and division and influx}.
    By denoting $\lambda_\infty = \max\left\{\|\lambda^\text{react}_1\|_{\infty}, \dots, \|\lambda^\text{react}_r\|_{\infty}, \|\lambda^\text{death}\|_\infty, \|\bar\lambda^\text{div}\|_\infty, \|\lambda^\text{in}\|_\infty \right\}$, these positive constants are given by $C_1 = 2 \sum_{z:\,\lambda^\text{in}(z)\neq 0} \lambda^2_{\infty} t^2 e^{2 \lambda_\infty t}$, 
    $C_2^z = \frac{\mu^2(z)}{2 \|\mu\|^2_1} e^{-2(r+2)\lambda_\infty t}$, and 
    $$ C_3 = 2 \|\mu\|^2_1 \left(1 + \sum_{z: \, \lambda^\text{in}(z)\neq 0} \frac{\lambda_\infty}{\mu(z)} \right)^2e^{2(r+2)\lambda_\infty t}
    + 
    \sum_{z: \, \lambda^\text{in}(z)\neq 0}
     \frac{C_1}{4} \left(\frac{\|\mu\|_1}{\mu(z)}\right)^3 e^{3(r+2)\lambda_{\infty} t}. $$
\end{theorem}
\begin{proof}
    {The proof is in the appendix}
\end{proof}

\begin{remark}
    Compared with \eqref{eq solution of the adjoint equation reaction, death, division}, the Feynman-Kac-like representations \eqref{eq. solution of n_t general case} and \eqref{eq. weights of the final algorithm} contain an additional term, $\int_{0}^{t} e^{\int_{s}^t \lambda^{D}\left(\tilde X(\tau)\right) \dd \tau}
        \frac{\lambda^\text{in}\left(\tilde X(s)\right)}{\|\mu\|_1\,\tilde p(t, \tilde X(s))}
        \dd s$, in the weight.
    From a biological perspective, this term compensates for the inability of $\tilde X(t)$ to model cell influx. 
\end{remark}

\begin{remark}
    The advantages of our FKG-SSA over 
    the agent-based simulation approach lie in three aspects.
    First, it only simulates a fixed number of modified cell systems $\tilde X(t)$, avoiding 
    the need to exhaustively track all the individual cells.
    Second, the simulation of $\tilde X(t)$ can be easily parallelized and has a linear computational complexity with the sample size $N$, significantly improving the
    agent-based method's 
    quadratic complexity with respect to the population size (see \Cref{section current challenges}).
    This implies that even when the FKG-SSA uses a sample size equivalent to the actual cell population, it remains considerably more efficient. 
    Third, while the agent-based method requires repeated simulation of the whole cell population $Y_t$ to achieve an accurate estimate of $n_t$ (according to the law of large numbers), our approach does not have this additional layer of computation, further enhancing efficiency.
\end{remark}

\subsection{Algorithm with resampling}
\label{section algorithm + restarting}

Essentially, our FKG-SSA is an importance sampling algorithm, where the samples are generated from the distribution of $\tilde X(t)$ for computing another (unnormalized) distribution $n_t$.
Thus, the accuracy of our method can be evaluated by the effective sample size (ESS), defined by $\text{ESS}\triangleq {\left( \sum_{i=1}^N w_i(t)\right)^2}/{\left( \sum_{i=1}^N w^2_i(t) \right)}$, which is commonly used for importance sampling algorithms \cite{liu2001monte}.
The ESS provides an interpretation that the importance sampling algorithm is approximately equivalent to the inference based on the ESS perfect samples from the target distribution \cite{doucet2009tutorial}.
Particularly when the ESS is small, only a few samples have large weights, and the final estimate depends almost only on these few samples, leading to inaccurate results. 
In contrast, when the ESS is large, all the samples have similar weights and contribute equally to the final estimate. 

In our algorithm, the ESS can decrease rapidly in time. 
To tackle this problem, we additionally introduce the resampling strategy to our algorithm. 
Resampling is a common approach to address the small ESS problem (also known as weight degeneracy) \cite{doucet2009tutorial,bain2009fundamentals}, which works by recursively replicating large-weight samples and discarding those with small weights in time.  
Tailored to our approach, since the FKG-SSA (\Cref{algorithm death and division and influx}) has a sampling step at the beginning, resampling is equivalent to periodically restarting the algorithm with the output $\hat n_t$ from the end of each time period serving as the initial condition for the next period.
The FKG-SSA with this resampling/restarting strategy is provided in \Cref{algorithm resampling}.
Notably, to ensure that \Cref{assumption 2} holds in each time period, we add the term $\lambda^\text{in}(x)/N$ to the output of the previous step when preparing the initial condition for the next time period (see Line 4 in \Cref{algorithm resampling}).
Since all the weights are equal at the beginning of each time period and the weights should not diverge significantly within a relatively small time interval, this strategy can control the ESS, thereby ensuring improved accuracy. 

\begin{algorithm2e}
\SetAlgoLined
Input the initial condition $\mu(x)$,  sample size $N$, final time $T$, and an increasing series of restarting times $t_1, \dots, t_M$ ($\in[0,T]$). \\
Denote $t_0 = 0$ and $t_{M+1}= T$. \\
\For{$k=0, \dots, M$}{
    $\bar \mu(x) = \mu(x)$ if $k=0$, else $\bar \mu(x) = \hat n_{t_k}(x) + \frac{\lambda^\text{in}(x)}{N}$. \hfill \tcp{Set initial conditions} 
    Run \Cref{algorithm death and division and influx} with initial condition $\bar \mu(x)$, the sample size $N$, and the final time $t_{k+1}-t_{k}$. Denote its solution by $\tilde n_t$ for $t\in(0,t_{k+1}-t_k].$ \\
    $\hat n_t = \tilde n_{t-t_k}$ for $t\in (t_k, t_{k+1}]$.
    \hfill\tcp{Approximated solution in $(t_k, t_{k+1}]$}
}
\caption{FKG-SSA with resampling/restarting}
\label{algorithm resampling}
\end{algorithm2e}

This refined algorithm is also convergent (see \Cref{them convergence of the FKG-SSA with resampling}), demonstrating its reliability.
The proof employs a relatively complicated math induction scheme due to the intricate dependence between the one-step error and the initial condition (e.g., see \Cref{them convergence of the FKG-SSA}).
We devote the entire \Cref{section proof the convergence of the FKG-SSA with resampling} to this proof. 

\begin{theorem}\label{them convergence of the FKG-SSA with resampling}
    Assume \Cref{Assumption 1}---\ref{assumption 2} hold, and let $T$ be the final time in \Cref{algorithm resampling}.
    Then, for any $t\in[0,T]$, there exist positive constants $\bar C_1$ and $\bar C_2$ (for $z$ in the support of $\lambda^\text{in}$) such that the output of \Cref{algorithm resampling} (denoted by $\hat n_t$) satisfies
    \begin{align*}
        \mathbb E\left[
            \|\hat n_t - n_t\|^2_2
        \right] ~\leq ~
         \bar C_1 N^2
        e^{-\bar C_2 N}
        \,+\,
        \frac{\bar C_1}{N}  
    \end{align*}
    where $n_t$ is the solution of \eqref{eq n_t dynamics}, and $N$ is the sample size of \Cref{algorithm death and division and influx}.
\end{theorem}
\begin{proof}
    The proof is provided in \Cref{section proof the convergence of the FKG-SSA with resampling}.
\end{proof}

\section{Numerical Examples}\label{section numerical examples}

We present a couple of biologically relevant numerical examples to illustrate our method. 
The algorithms were performed on a server equipped with a 2.70GHz Intel Xeon Platinum 8270 CPU.

\subsection{Transcriptional feedback model of protein expression}

We consider a transcriptional feedback model where protein synthesis inhibits the cell division rate and promotes its own production rate \cite{piho2024feedback}. 
In this model, each cell has a chemical species $S_1$ (protein) and undergoes two reactions: the protein production ($\emptyset \to S_1$) and protein degradation ($S_1 \to \emptyset$).
Their propensities are respectively $\lambda^{\text{prod}}(x) = \alpha+{k_1}/{\left[\left(\frac{K_1}{x}\right)^{2} +1\right]}$ 
and $\lambda^{\text{deg}}(x) = \delta_\text{deg}\, x$.
Moreover, all the cells can divide and die with a protein-dependent division rate $\lambda^\text{div}(x, x', x'') = \frac{x!}{x'!x''!}\frac{1}{2^x}{k_2}/{\left[(\frac{x}{K_2})^{4} +1\right]} $ (for $x=x'+x''$)
and a constant death rate $\lambda^\text{death}(x)=\delta_\text{death}$. 
Finally, we assume that these cells originate from stem cells with internal state $x=0$, which manifests as cell influx $\lambda^\text{in}(x) = \delta_\text{in}\, \mathbbold 1_{0}$.
Here, $\alpha$, $k_1$, $K_1$, $k_2$, $K_2$, $\delta_\text{deg}$, $\delta_\text{death}$, and $\delta_\text{in}$ are system parameter whose values are chosen as $\alpha = 588$, $k_1 = 5600$, $K_1 = 140$, $k_2 = 40$, $K_2 = 16.46$, $\delta_\text{deg}=25$, 
$\delta_\text{death}=1$, and $\delta_\text{in}=10$.
A graphic illustration of this system is given in \Cref{Fig example 1}.A.
In this example, we let the system initially contain ten cells with internal states $x=0$, and we intend to compare our method with the classical agent-based method.

The numerical results are presented in \Cref{Fig example 1}.
Here, we use the solution of the FSP with a truncated state space $\{0, 1, \dots, 50\}$ as the benchmark (see the blue curves in \Cref{Fig example 1}.C).
Since this truncated space is relatively small, the FSP is efficient with a CPU time of less than 1 minute. 
In contrast, the agent-based algorithm is significantly more inefficient, requiring a CPU time of 27 days to simulate 50 samples for a reasonably good estimation (\Cref{Fig example 1}.C). 
As shown in \Cref{Fig example 1}.B, our FKG-SSAs are far more efficient than the agent-based simulation model; to achieve the same level of accuracy, our methods use only about 1/1000 of the time.
Additionally, \Cref{Fig example 1}.B shows that the squared errors of our approaches decay at the rate of $\frac{1}{\text{CPU Time}}$ (or equivalently ${1}/{N}$), which aligns with our theoretical results in \Cref{them convergence of the FKG-SSA} and \Cref{them convergence of the FKG-SSA with resampling}.
Furthermore, for the same time cost, the FKG-SSA utilizing resampling is more accurate than the FKG-SSA without resampling (\Cref{Fig example 1}.B), due to the larger ESS achieved by the former (\Cref{Fig example 1}.C).
Overall, the results demonstrate the superior performance of our algorithm over the agent-based method and support all our theoretical results.

\begin{figure}[htbp]
	\centering
	\includegraphics[width=1\linewidth]{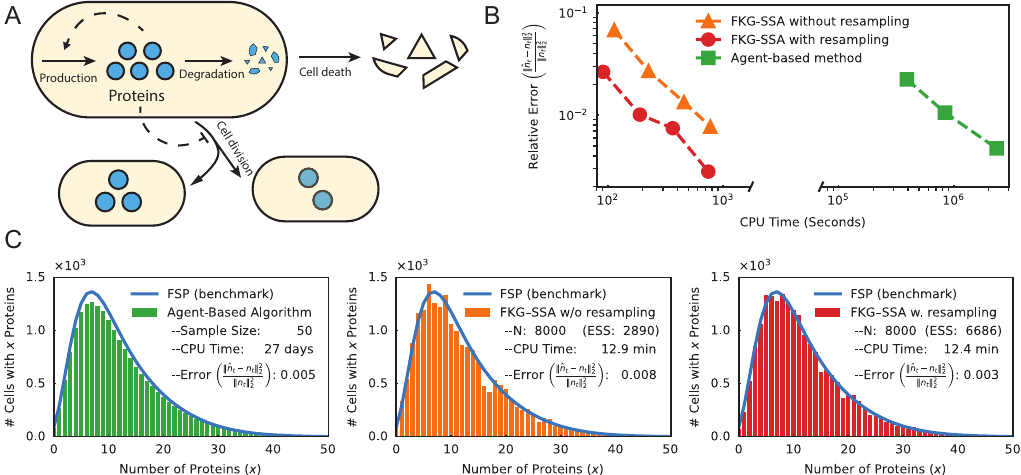}
	\caption{\textbf{Protein expression model.} (A) Graphic illustration of the system. (B) Comparison of different approaches to simulating the mean dynamics  \eqref{eq n_t dynamics}. 
    The plot shows their performance in approximating the solution at $T=0.25$.
    The solution obtained from the finite-state projection (FSP) with the truncated state space $\{x \in\mathbb Z_{\geq 0}| 0\leq x\leq 50\}$ is used as the benchmark. 
    The FKG-SSA with resampling performs resampling every 0.05 units of time.
    (C) Approximated solutions of \eqref{eq n_t dynamics} at time $T=0.25$. }
	\label{Fig example 1}
\end{figure}

\subsection{Cancer-immune co-evolution model of mutation accumulation}
Next, we consider a cancer-immune co-evolution model \cite{chen2024frequency, lakatos2020evolutionary} where cancer cells accumulate antigenic and immune-escape mutations that affect cell death rates.
A graphic illustration of this model is given in \Cref{Fig example 2}.A.
In this example, each cell has a 3-dimensional state vector $x = (x_1, x_2, x_3)^{\top}$, representing the accumulated antigenic mutations, accumulated antigenicity, and the presence of immune-escape mutations. 
The accumulation of antigenic mutations results in increased antigenicity, which raises the chance of the cell being recognized by the immune system and subsequently leads to a higher cell death rate.
In contrast, the presence of immune-escape mutations protects the cell from immune system attacks, allowing for a decreased death rate.
Consequently, the death rate is modeled by $\lambda^\text{death}(x) = k^\text{death}_1 + k^\text{death}_2 x_2 (1-x_3)$, where $k^\text{death}_1$ represents the basal death rate, and $k^\text{death}_2 x_2 x_3$ accounts for additional death rate caused by the attack from the immune system.
In addition, all the cells have a constant division rate $\bar \lambda^\text{div}(x) = k^\text{div}$, which is independent of their internal states. 
At each division, the daughter cells can acquire additional antigenic mutations, with the number of mutations following a Poisson distribution $\text{Pois}(k_{A})$.
The antigenic mutations exhibit significant variability in their associated antigenicity \cite{luksza2017neoantigen}.
To capture this variability, the antigenicity of each newly accumulated antigenic mutation is modeled using a geometric distribution with mean $1/p_{A}$. 
Thus, given the number of newly accumulated antigenic mutations (denoted by $\Delta x_1$) in a daughter cell, its newly obtained antigenicity follows the negative binomial distribution $\text{NB}(\Delta x_1, p_A)$.
Additionally, the daughter cells acquire the immune-escape mutation (if their mother cell does not have it) with a probability $p_\text{IE}$.
With these notations, the probability of obtaining two daughter cells of internal states $x'$ and $x''$ (in a specific order) given the mother cell's internal state $x$ is expressed as 
\begin{align*}
    p^\text{div}(x',x''|x)
    =& \mathbb P_\text{Pois}(x'_1 -x_1)
    \, \mathbb P_\text{NB}(x'_2 -x_2 | x'_1 -x_1)\, p_\text{IE}^{x'_3\, \mathbbold 1_0(x_3)} (1-p_\text{IE})^{(1-x'_3) \mathbbold 1_0(x_3)} \\
    & \times 
    \mathbb P_\text{Pois}(x''_1 -x_1)
    \, \mathbb P_\text{NB}(x''_2 -x_2 | x''_1 -x_1)\, p_\text{IE}^{x''_3\, \mathbbold 1_0(x_3)} (1-p_\text{IE})^{(1-x''_3) \mathbbold 1_0(x_3)}
\end{align*}
where $\mathbb P_\text{Pois}$ is the Poisson distribution with mean $k_A$, and $\mathbb P_\text{NB}(\cdot | \Delta x_1)$ is the negative binomial distribution with the number of successes $\Delta x_1$ and the success probability $p_A$.
Therefore, the division/mutation rate $\lambda^\text{div}(x,x',x'')$ has the form 
$\lambda^\text{div}(x,x',x'') 
    = \bar \lambda^\text{div}(x) \,p^\text{div}(x',x''|x)
    = k^\text{div}\,p^\text{div}(x',x''|x).$
In this model, we let the system initially contain {10} cells with internal states $x=(0,0,0)^{\top}$ and choose the parameters as {$k^\text{div} = 0.5$, $k^\text{death}_1 = 0.1$, $k^\text{death}_2 = 
0.072$, $k_A = 0.5$, $p_A = 1/3$, and $p_\text{IE} = 10^{-4}$}.
Here, we aim to compare our method with the classical agent-based method and the FSP approach.

\begin{figure}[htbp]
	\centering
	\includegraphics[width=1\linewidth]{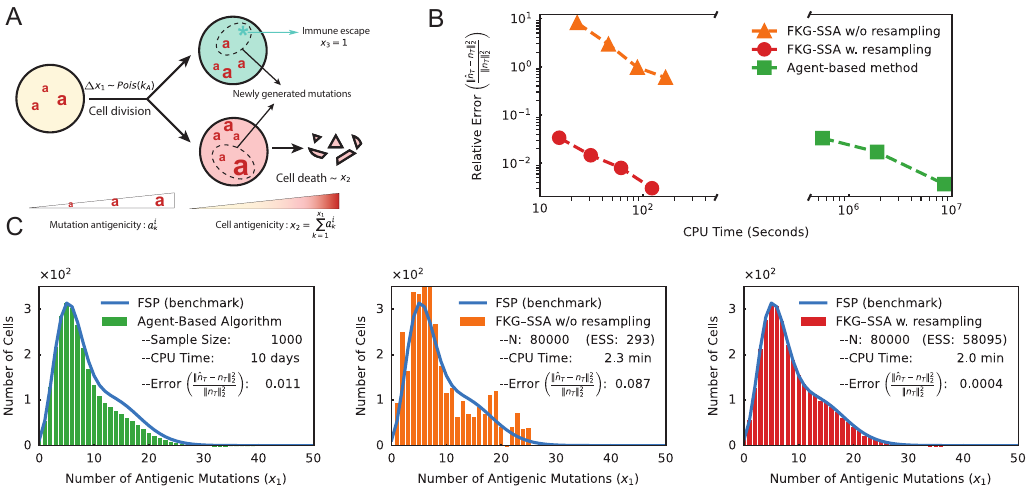}
	\caption{\textbf{Cancer-immune co-evolution model.}  (A) Graphic illustration of the co-evolution system. 
    (B) Comparison of different approaches to simulating the mean dynamics  \eqref{eq n_t dynamics}. 
    The plot shows their performance in approximating the solution at $T=30$.
    The solution obtained from the finite-state projection (FSP) with the truncated state space $\{0,1, \dots, 50\} \times \{0,1, \dots, 200\} \times \{0, 1\}$ is used as the benchmark. The CPU time for the FSP was around 11 hours. 
    The FKG-SSA with resampling performs resampling every 3 units of time.
    (C) Approximated marginal solutions of \eqref{eq n_t dynamics} at time $T=30$.}
	\label{Fig example 2}
\end{figure}

The numerical results are presented in \Cref{Fig example 2}.
Again, we use the solution of the FSP (with a truncated state space $\{0,1, \dots, 50\} \times \{0,1, \dots, 200\} \times \{0, 1\}$) as the benchmark. 
In this example, the size of the truncated state space exceeds $2\times 10^4$, resulting in the FSP taking approximately 11 hours to compute.
The agent-based simulation algorithm is even less efficient, requiring a CPU time of 10 days to simulate 1000 samples for an accurate result (see \Cref{Fig example 2}.B and \Cref{Fig example 2}.C).
In contrast, our approaches are significantly more efficient than these two conventional approaches (see \Cref{Fig example 2}.B and \Cref{Fig example 2}.C).
Specifically, the FKG-SSA with resampling only needs 1/10,000 of the time required by the agent-based algorithm to achieve the same level of accuracy (\Cref{Fig example 2}.B). 
Without resampling, the FKG-SSA can lead to an extremely small effective sample size (ESS); see the middle panel in \Cref{Fig example 2}.C where the ESS is only 293 despite simulating 80,000 trajectories.
This reduction in ESS renders the FKG-SSA without resampling less accurate compared to the one with resampling (\Cref{Fig example 2}.B and \Cref{Fig example 2}.C), which underscores the necessity of having the resampling strategy in the FKG-SSA framework.  
Moreover, \Cref{Fig example 2}.B illustrates that the squared errors of our approaches decrease at the rate of $\frac{1}{\text{CPU Time}}$ (or equivalently ${1}/{N}$), which agree with our theoretical results \Cref{them convergence of the FKG-SSA} and \Cref{them convergence of the FKG-SSA with resampling}.
In summary, this example demonstrates the superior performance of our algorithms and confirms the validity of our theoretical results.

\section{Conclusion}\label{section conclusion}

Efficient simulation of growing cell populations is crucial for investigating living biological systems. 
Despite advancements over the past few decades, state-of-the-art approaches still face inefficiencies due to the increasing number of cells or the high dimensionality of the system. 
To address these challenges, this paper proposed a novel fixed-budget simulation approach based on a Keynman-Kac-like representation of the population dynamics (see \Cref{them solution of n_t general case}). 
This approach, named the Feynman-Kac-inspired Gillespie's  Stochastic Simulation Algorithm (FKG-SSA), always employs a fixed number of parallelly simulated cells for Monte Carlo computation of the system, leading to a constant computational complexity regardless of the actual cell population size. 
Furthermore, this method has guaranteed convergence properties (see \Cref{them convergence of the FKG-SSA} and \Cref{them convergence of the FKG-SSA with resampling}), demonstrating its accuracy and reliability in practical applications. 
We have also illustrated its superior performance with a couple of biologically relevant numerical examples.
Overall, the proposed FKG-SSA effectively addresses the challenges of simulating growing cell populations, providing a solid foundation for better analysis of biological systems. 

There are numerous topics deserving further exploration in future work.
First, the whole framework can be further extended to include more physiological mechanisms, such as cellular spatial motions, quorum sensing, and cell-to-cell interactions, all of which are essential for practical applications. 
Second, model reduction approaches, such as the diffusion model \cite{gillespie2000chemical} and the tau-leaping method \cite{gillespie2001approximate}, can also be incorporated into the framework to further improve the efficiency. 
Third, the efficiency of our approach underscores its good potential for uncovering key biological mechanisms from experimental data. 
Nevertheless, effectively integrating biological data with our method remains an open problem requiring further investigation.
Furthermore, the current FKG-SSA only computes the mean dynamics of the cell population $Y_t$.
Developing similar approaches for analyzing the variance dynamics of $Y_t$ would provide deeper insights into the biological system.
In summary, this paper opens the door to a range of compelling problems.

\renewcommand{\thesubsection}{\Alph{section}.\arabic{subsection}}
\renewcommand{\theproposition}{A\arabic{proposition}}

\appendix

\section{Proofs of the results}

Here, we provide the proofs of the results in this paper.
We dedicate the entire \Cref{section proof the convergence of the FKG-SSA with resampling} to the relatively complicated proof of \Cref{them convergence of the FKG-SSA with resampling}.

\begin{proof}[\textbf{Proof of \Cref{proposition non-explosivity}}]
    Under \Cref{Assumption 1}, $Y_t$ is always a non-negative integer-valued function up to the explosion time.
    Let $\tau_c = \min\{t~|~ \|Y_t\|_{1} \geq c\}$.
    By Dynkin's formula \cite[p. 133]{dynkin1965markov},  we have 
    \begin{align*}
        \mathbb E\left[\|Y_{t\wedge \tau_c }\|_1\right]
        &= \|Y_0\|_1 + \sum_{x\in\mathbb S}\mathbb E\left[\int_{0}^{t\wedge \tau_c}
        Y_s(x)\left(
        \sum_{x',x''\in\mathbb S} \lambda^\text{div}(x,x',x'')
        -\lambda^\text{death}(x) 
        \right)
        + \lambda^\text{in}(x)
        ~\dd s
        \right]
    \end{align*}
    By the boundedness of $\lambda^\text{react}_j(x)$, $\lambda^\text{death}(x)$, 
    $\lambda^{in}(x)$,
    $\bar \lambda^\text{div}(x)$, and $\|\lambda^{in}(x)\|_1$
    (\Cref{Assumption 1}),  there exists a positive constant $C$ such that 
    \begin{align*}
        \mathbb E\left[\|Y_{t\wedge \tau_c }\|_1\right]
        \leq \|Y_0\|_1 +  
         C \mathbb E\left[\int_{0}^{t\wedge \tau_c}
        \|Y_s(x)\|_1 + 1 \dd s \right]
        \leq \|Y_0\|_1 +  
         C \int_{0}^{t}
        \mathbb E\left[\|Y_{s\wedge \tau_{t_c}}(x)\|_1 \right]+ 1 \dd s .
    \end{align*}
    By Gronwall’s inequality, we have   
    \begin{align}\label{eq. boundedness of the L1 norm of Y}
        \mathbb E\left[\|Y_{t\wedge \tau_c }\|_1\right] \leq 
        \left(\|Y_0\|_1 + Ct \right)e^{Ct}
        && \forall t\geq 0.
    \end{align}
    Furthermore, by Markov's inequality, we have that for $t\geq 0$,
    \begin{align*}
        \lim_{c\to\infty} \mathbb E\left[\mathbbold 1\left(\tau_c<t\right)\right]
        ~=~
        \lim_{c\to\infty} \mathbb P(\tau_c<t)
        ~=~ \lim_{c\to\infty} \mathbb P\left(\|Y_{t\wedge \tau_c}\|_1 \geq c\right)
        ~\leq~ \lim_{c\to\infty} \frac{\mathbb E\left[\|Y_{t\wedge \tau_c }\|_1\right]}{c} ~=~ 0
    \end{align*}
    where the last equality follows from \eqref{eq. boundedness of the L1 norm of Y}.
    The dominated convergence theorem suggests
    \begin{align}\label{eq. tau infity smaller than any t}
        \mathbb P\left(\lim_{c\to\infty} \tau_c < t \right)
        ~=~ \mathbb E\left[
        \lim_{c\to\infty}\mathbbold 1\left(
            \tau_c<t
        \right)
        \right]
        ~=~ \lim_{c\to\infty} \mathbb E\left[
        \mathbbold 1\left(
            \tau_c<t
        \right)
        \right]
        ~=~0
        && \forall t\geq 0.
    \end{align}
    We define $\tau_{\infty} = \lim_{c\to\infty} \tau_c$.
    Notice that $\{ \tau_\infty<t \text{ for some $t>0$}\} = \limsup_{N\to\infty} \{\tau_\infty< N \}$ with $N$ being integers. 
    Also, we have that $\sum_{N=1}^\infty \mathbb P\left(\tau_{\infty} < N \right)
    = 0 $ due to \eqref{eq. tau infity smaller than any t}.
    Then, by the Borel-Cantelli lemma, we have $\mathbb P \left(\tau_\infty<t \text{ for some $t>0$} \right) = 0$, which implies $\tau_\infty = + \infty$ almost surely (i.e., non-explosivity).
    Finally, by applying the Fatou's lemma to \Cref{eq. boundedness of the L1 norm of Y}, we have $\mathbb E\left[\|Y_t\|_1\right] \leq \lim_{c\to c}\mathbb E\left[\|Y_{t\wedge \tau_c }\|_1\right] \leq 
        \left(\|Y_0\|_1 + Ct \right)e^{Ct}< +\infty$ for any $t\geq 0$.
\end{proof}

\begin{proof}[\textbf{Proof of \Cref{proposition uniqueness of the mean dynamics}}]
    By \Cref{proposition non-explosivity}, $\mathbb E\left[Y_t\right]$ exists at any time $t$.
    By straightforward calculation, we can show that it solves \eqref{eq n_t dynamics}. 
    Then, we only need to show the uniqueness. 
    Let $\rho_1(t)$ and $\rho_2(t)$ be any two
    solutions of \eqref{eq n_t dynamics} that also satisfy  $\rho_1(t),\rho_2(t) \in  L^1\left( \mathbb Z^{d}_{\geq 0}\right)$ for all $t\geq 0$.
    By denoting $\rho_3(t) = \rho_1(t) - \rho_2(t)$, we have that for all $t\geq 0$,  $\rho_3(t) \in  L^1\left( \mathbb Z^{d}_{\geq 0}\right)$ and 
    $\frac{\dd }{\dd t}\rho_3(t)
        = \left(A_\text{react} - \lambda^\text{death} + A_\text{div}\right) \rho_3(t)$.
    Based on \Cref{Assumption 1}, the operator $\left(A_\text{react} - \lambda^\text{death} + A_\text{div}\right)$ from $ L^1\left( \mathbb Z^{d}_{\geq 0}\right)$ to itself is bounded.
    Therefore, there exists a constant $C$ such that  $\frac{\dd}{\dd t}\left\|
            \rho_3(t)
        \right\|_1
        \leq C \left\|
            \rho_3(t)
        \right\|_1
        $ for all $ t\geq 0$.
    Notice that $\left\|
            \rho_3(0)
        \right\|_1=0$.
    Then, by Gronwall's inequality, we can conclude that $\left\|
            \rho_3(t)
        \right\|_1 = 0$ for all $t\geq 0$,
        implying $\rho_1(t) = \rho_2(t)$ for all $t\geq 0$.
    This proves the uniqueness of the solution.
\end{proof}

\begin{proof}[\textbf{Proof of \Cref{proposition adjoint property reaction + death}}]
    The only thing we need to prove is that the proposed $u_t$ satisfies $\frac{\dd}{ \dd t} u_t = \left(A_\text{react} - \lambda^\text{death} \right) u_t$.
    By the invariance of
    $\left<\phi^{T,g}_t, u_t\right>$ over time, we have  
    \begin{align*}
       0 &=  \left< \phi^{T, g}_{t_1} \,,\,  u_{t_1}\right> - \left< \phi^{T, g}_{t_2} \,,\,  u_{t_2}\right> \\
       & = \left< \phi^{T, g}_{t_1} - \phi^{T,g}_{t_2} \,,\,  u_{t_1}\right>
        - \left< \phi^{T, g}_{t_2}\,,\,  u_{t_2} - u_{t_1}\right> \\
        & = \left< \text{e}^{-\left(A^*_\text{react} - \lambda^\text{death} \right)(t_{1}-t_{2})} \phi^{T,g}_{t_2}\,,\,  u_{t_1}\right>
        - \left< \phi^{T, g}_{t_2} \,,\,  u_{t_2} - u_{t_1}\right>
        && \text{($\phi^{T,g}_{t}$ solves \eqref{eq. adjoint equation reaction + death})} \\
        & = \left< \phi^{T,g}_{t_2}\,,\,  \text{e}^{-\left(A_\text{react} - \lambda^\text{death} \right)(t_{1}-t_{2})}u_{t_1} - u_{t_2} - u_{t_1} \right>
        && \text{(adjoint operators)}
    \end{align*}
    for any $T>0$, $0\leq t_1< t_2 \leq T$, and bounded function $g$.
    By the arbitrariness of $T$ and $g$, we have that $\text{e}^{-\left(A_\text{react} - \lambda^\text{death} \right)(t_{1}-t_{2})}u_{t_1} - u_{t_2} - u_{t_1} = 0$ for any $t_2>t_1\geq 0$.
    Thus, we have
    \begin{align*}
        \frac{\dd}{\dd t}
        u_{t_1} 
        = \lim_{t_2\to t_1^+} \frac{u_{t_2} - u_{t_1}}{t_2 -t_1}
        =  \lim_{t_2\to t_1^+} \frac{\text{e}^{-\left(A_\text{react} - \lambda^\text{death} \right)(t_{1}-t_{2})}u_{t_1}}{t_2 -t_1}
        = \left(A_\text{react} - \lambda^\text{death} \right)u_{t_1},
        && \forall t_1\geq 0,
    \end{align*}
    which proves the result. 
\end{proof}

\begin{proof}[\textbf{Proof of \Cref{them solution of n_t reaction + death}}]
    The uniqueness of the solution is guaranteed by \Cref{proposition uniqueness of the mean dynamics}.
    Moreover, the process $n_t$ defined by \eqref{eq. solution of n_t reaction + death} is time differentiable and satisfies $n_0 = \mu$.
    Then, we only need to check whether the value of $<\phi^{T,g}_t, n_t>$ is unchanged over time for any $T>0$ and bounded function $g$.
    Based on \eqref{eq. solution of n_t reaction + death}, we have
    \begin{align*}
        &<\phi^{T,g}_t, n_t>\\
        & =\|\mu\|_1 
        \mathbb E\left[
        \sum_{x\in\mathbb Z_{\geq 0}^d}  
        \phi^{T,g}_t (x)\,
        \mathbbold 1_{x} \left( X(t) \right)\,
        \exp\left(\int_{0}^{t}-\lambda^{\text{death}}(X(s))ds\right)
        \right]
        \qquad \text{(Fubini's theorem)}\\
        &=\|\mu\|_1 
        \mathbb E\left[ 
        \phi^{T,g}_t (X(t))\,
        \exp\left(\int_{0}^{t}-\lambda^{\text{death}}(X(s))ds\right)
        \right] \\
        &=\|\mu\|_1 
        \mathbb E\left[ 
        \mathbb E\left[ g(X(T))e^{\int_{t}^{T}-\lambda^{\text{death}}(X(s))ds} \bigg| \mathcal F_t\right]
        \exp\left(\int_{0}^{t}-\lambda^{\text{death}}(X(s))ds\right)
        \right]
        \qquad \text{(by \eqref{eq solution of the adjoint equation reaction + death})}\\
        & = \|\mu\|_1 
        \mathbb E\left[ g(X(T))
        \mathbb E\left[
        \exp\left(\int_{0}^{T}-\lambda^{\text{death}}(X(s))ds\right)
        \bigg| \mathcal F_t\right]
        \right] \\
        & = \|\mu\|_1 
        \mathbb E\left[ g(X(T))
        \exp\left(\int_{0}^{T}-\lambda^{\text{death}}(X(s))ds\right)
        \right]
        \qquad\qquad
        \text{(Law of total expectation)}
    \end{align*}
    for any $T>0$, $t\in[0,T]$, and bounded function $g$.
    This suggests that the value of the inner product $<\phi^{T,g}_t, n_t>$ does not change over time, which proves the result. 
\end{proof}

\begin{proof}[\textbf{Proof of \Cref{proposition non-zero of the probability}}]
    For any $z$ in support of $\mu(x)$ and any $t>0$, we have 
    \begin{align*}
        \mathbb P\left(\tilde X(t) = z\right) 
        ~\geq~
        \mathbb P\left(\tilde X(s) = z, \forall s \in [0,t] \right)
        ~=~ e^{-t\left(\sum_{j=1}^n\lambda^\text{react}(z) + 2 \bar \lambda^\text{div}(z)\right)} P\left(\tilde X(0) = z\right).
    \end{align*}
    This, together with the boundedness of the rate functions (see \Cref{Assumption 1}), proves the result for any $z$ in support of $\mu(x)$.
    Since the support of $\lambda^\text{in}(x)$ is a subset of the support of $\mu(x)$, this also proves the result for the states within the support of $\lambda^\text{in}(x)$.
\end{proof}

\begin{proof}[\textbf{Proof of \Cref{them convergence of the FKG-SSA}}]
    Following the techniques in \cite{antonelli2002rate}, we construct new weights (denoted as $\tilde w_i(t)$) for the samples in \Cref{algorithm death and division and influx} by solving the differential equation
    \begin{align*}
       \frac{\dd }{\dd t} \tilde w_i(t)
       = \left(\bar \lambda^\text{div}\left(\texttt{x}_i(t)\right) - \lambda^\text{death}\left(\texttt{x}_i(t)\right)\right) \tilde w_i(t)
       + \frac{\lambda^\text{in}(\texttt{x}_i(t))}{\|\mu\|_1\,\tilde p(t, \texttt{x}_i(t))}
       && \text{and } \tilde w(0) = 1.
    \end{align*}
    Since the rate functions are bounded \Cref{Assumption 1}, each of these ODEs has a unique solution. 
    Now, we establish a new estimate of the solution $n_t$ by $\tilde n_t(x) = \frac{\|\mu\|_1}{N}\sum_{i=1}^{N} \mathbbold 1_{x}(\texttt{x}_{i}(t))\tilde w_{i}(t)$.
    In the following, we analyze the convergence of \Cref{algorithm death and division and influx} based on the inequality
    \begin{align}\label{eq. for comparing the convergence of FKG-SSA}
        \mathbb E\left[
            \|\hat n_t - n_t\|^2_2
        \right]
        ~\leq ~
        2\,\mathbb E\left[
            \|\hat n_t - \tilde n_t\|^2_2
        \right]
        + 
        2\,\mathbb E\left[
            \| \tilde n_t - n_t\|^2_2
        \right].
    \end{align}

    First, we focus on $\mathbb E\left[
            \| \tilde n_t - n_t\|^2_2
        \right]$.
    Notice that we use the exact distribution $\tilde p(t, x)$ when constructing $\tilde w_i(t)$.
    Therefore, $\left(\texttt{x}_{1}(t), \tilde w_{1}(t)\right), \dots, \left(\texttt{x}_{N}(t), \tilde w_{N}(t)\right)$ are independent and identically distributed according to the distribution of $\left(W_t, \tilde X(t)\right)$. (Recall that $\tilde X(t)$ and $W_t$ are defined in \eqref{eq tilde X} and \eqref{eq. weights of the final algorithm}, respectively).
    Based on \Cref{assumption 2} and \Cref{proposition non-zero of the probability}, we have the relation $\frac{\lambda^\text{in}(x)}{\|\mu\|_1\,\tilde p(t, x)} \leq \frac{\lambda_\infty}{\mu(x)} e^{(r+2)\lambda_\infty t}$ for all $x$ in the support of $\lambda^\text{in}$;
    here $\lambda_\infty = \max\left\{\|\lambda^\text{react}_1\|_{\infty}, \dots, \|\lambda^\text{react}_r\|_{\infty}, \|\lambda^\text{death}\|_\infty, \|\bar\lambda^\text{div}\|_\infty, \|\lambda^\text{in}\|_\infty \right\}$.
    Then, applying this to \eqref{eq. weights of the final algorithm}, we further have $0<W_t< \left(1 + \sum_{\{x: \, \lambda^\text{in}(x)\neq 0\}} \frac{\lambda_\infty}{\mu(x)} \right)e^{(r+2)\lambda_\infty t}$ almost surely.
    This result suggests that the random variable $\mathbbold 1_x\left(\tilde X(t)\right) W_t $ is almost surely bounded. 
    By combining this boundedness with the independence and identical distributedness of $\left(\texttt{x}_{1}(t), \tilde w_{1}(t)\right), \dots, \left(\texttt{x}_{N}(t), \tilde w_{N}(t)\right)$, we have $\mathbb E\left[ \left(\tilde n_t(x) - n_t(x)\right)^2 \right]
      = \|\mu\|^2_1 \frac{\text{Var}\left(\mathbbold 1_x\left(\tilde X(t)\right) W_t  \right) }{N}$
    and furthermore
    \begin{align*}
      \frac{\mathbb E\left[ \left(\tilde n_t(x) - n_t(x)\right)^2 \right]}{\|\mu\|^2_1}
      \leq \frac{\mathbb E\left( \mathbbold 1_x\left(\tilde X(t)\right) W_t^2  \right) }{N} 
      \leq \frac{\tilde p(t,x)\,\left(1 + \sum_{\{x: \, \lambda^\text{in}(x)\neq 0\}} \frac{\lambda_\infty}{\mu(x)} \right)^2e^{2(r+2)\lambda_\infty t}}{N}
    \end{align*}
    for any $x\in\mathbb Z^{d}_{\geq 0}$ and $t\geq 0$, where the first inequality follows from the fact that the second moment is no less than the variance.
    Then, by Fubini's theorem, we further conclude 
    \begin{align}\label{eq. convergence proof of FKG-SSA the second part}
         \mathbb E\left[
            \| \tilde n_t - n_t\|^2_2
        \right]
        = \sum_{x\in\mathbb Z^{d}_{\geq 0}} \mathbb E\left[ 
            \left(\tilde n_t(x)- n_t(x)\right)^2
        \right] 
        \leq \frac{\bar C_1}{N}
        && \forall t\geq 0
    \end{align}
    with $\bar C_1 = \|\mu\|^2_1 \left(1 + \sum_{\{x: \, \lambda^\text{in}(x)\neq 0\}} \frac{\lambda_\infty}{\mu(x)} \right)^2e^{2(r+2)\lambda_\infty t}$.

    Next, we analyze $\mathbb E\left[
            \|\hat n_t - \tilde n_t\|^2_2
        \right]$.
    First, by definition, we can obtain 
    \begin{align}\label{eq. bounds of the absolution value of hat n_t and tilde n_t}
        \left\|\hat n_t - \tilde n_t\right\|_2^2
        &= \sum_{x\in\mathbb Z^{d}_{\geq 0}}\left(\frac{\|\mu\|_1}{N} \sum_{i=1}^{N} \mathbbold 1_{x}(\texttt{x}_{i}(t))\big(w_{i}(t)-\tilde w_{i}(t)\big)\right)^2  \notag\\
        &\leq |\mu\|_1^2 \sum_{x\in\mathbb Z^{d}_{\geq 0}} \left( \sum_{i=1}^{N} \frac{ \mathbbold 1_{x}(\texttt{x}_{i}(t))}{N} \right)^2
        \left(\sup_{i}\left|w_{i}(t)-\tilde w_{i}(t)\right|\right)^2 \notag \\
        &\leq \|\mu\|_1^2
        \left(\sup_{i}\left|w_{i}(t)-\tilde w_{i}(t)\right|\right)^2
    \end{align}
    where the last line follows from H\"older's inequality $\|\hat p(t,\cdot) \|_2^2 \leq \|\hat p(t,\cdot) \|_1 \|\hat p(t,\cdot) \|_\infty = 1$.
    This then draws our attention to the quantity $\sup_{i}\left|w_{i}(t)-\tilde w_{i}(t)\right|$.
    Let's denote $\Delta w_i(t) = w_{i}(t)-\tilde w_{i}(t)$.
    By solving the differential equations associated with these weights, we can obtain 
    \begin{align}\label{eq. Delta w}
        \left(\Delta w_i(t)\right)^2 &= \left[\int_{0}^{t} e^{\int_{s}^t \lambda^{D}\left(\tilde X(\tau)\right) \dd \tau}
        \frac{\lambda^\text{in}\left(\texttt{x}_{i}(t) \right)}{\|\mu\|_1}
        \left(\frac{1}{\hat p(t, \texttt{x}_{i}(s))} - 
        \frac{1}{\tilde p(t, \texttt{x}_{i}(s))}\right)
        \dd s.\right]^2 \\
        & \leq 
        \frac{\lambda^2_{\infty} e^{2\lambda_\infty t}}{\|\mu\|_1^2}
        \bar C_2 t
        \int_{0}^{t} 
        \sum_{z:\,\lambda^\text{in}(z)\neq 0 \,\&\, \hat p(t,z)>0 }
        \left(\frac{1}{\hat p(t, z)} - 
        \frac{1}{\tilde p(t, z)}\right)^2
        \dd s \notag
    \end{align}
    where the second line follows from Jensen's inequality, and $\bar C_2$ is the 
    cardinality of the support of $\lambda^\text{in}$.
    By denoting two intervals $I^1_{t,z} = \left(0, \frac{1}{2}\tilde p(t,z)\right)$ and $I^2_{t,z} = \left[\frac{1}{2}\tilde p(t,z), 1\right]$, we also have 
    \begin{align}
        &\mathbb E\left[ 
        \mathbbold 1\big( \hat p(t, z) > 0 \big)
        \left(\frac{1}{\hat p(t, z)} - 
        \frac{1}{\tilde p(t, z)}\right)^2
        \right] \notag\\
        &=\mathbb E\left[ 
        \mathbbold 1\left( \hat p(t, z) \in I^{1}_{t,z} \right)
        \left(\frac{1}{\hat p(t, z)} - 
        \frac{1}{\tilde p(t, z)}\right)^2
        \right]  
        +
        \mathbb E\left[ 
        \mathbbold 1\left( \hat p(t, z) \in I^{2}_{t,z} \right)
        \left(\frac{1}{\hat p(t, z)} - 
        \frac{1}{\tilde p(t, z)}\right)^2
        \right] \notag\\
        &\leq N^2\, \mathbb P \left( \hat  p(t,z)- \tilde p(t,z)< -\frac{1}{2}\tilde p(t,z)  \right)
        + \mathbb E\left[ 
        \frac{\left(\tilde p(t, z) - \hat p(t, z)\right)^2}{4\left(\tilde p(t, z)\right)^4}
        \right] \notag\\
        &\leq N^2\, e^{-\frac{1}{2}\left(\tilde p(t,z)\right)^2 N}
        + \frac{1 - \tilde p(t, z)}{4N\left(\tilde p(t, z)\right)^3}
        \label{eq. second moments of 1/p}
    \end{align}
    where the third line follows from that $1 /\hat p(t, z)<N$ when $\hat p(t, z)> 0$, and the last line follows from Hoeffding's inequality \cite{hoeffding1963probability}. 
    Then, combining this with \eqref{eq. bounds of the absolution value of hat n_t and tilde n_t} and \eqref{eq. Delta w}, 
    we have 
    \begin{align*}
        \mathbb E\left[\left|\hat n_t - \tilde n_t\right|_2^2\right]
        &\leq \bar C_2 \lambda^2_{\infty} t e^{2 \lambda_\infty t} \, \mathbb E\left[
        \int_{0}^{t} 
        \sum_{z:\,\lambda^\text{in}(z)\neq 0 \,\&\, \hat p(s,z)>0 }
        \left(\frac{1}{\hat p(s, z)} - 
        \frac{1}{\tilde p(s, z)}\right)^2
        \dd s \right] \\ 
        &\leq  \bar C_2 \lambda^2_{\infty} t e^{2 \lambda_\infty t} 
        \sum_{z:\,\lambda^\text{in}(z)\neq 0}
        \int_{0}^{t}
        \mathbb E\left[
        \mathbbold 1\big( \hat p(s, z) > 0 \big)
        \left(\frac{1}{\hat p(s, z)} - 
        \frac{1}{\tilde p(s, z)}\right)^2
        \right]\dd s \\
        &\leq \bar C_2 \lambda^2_{\infty} t e^{2 \lambda_\infty t} 
        \sum_{z:\,\lambda^\text{in}(z)\neq 0}
        \int_{0}^{t}
        N^2\, e^{-\frac{1}{2}\left(\tilde p(s,x)\right)^2 N}
        + \frac{1 - \tilde p(s, z)}{4 N\left(\tilde p(s, z)\right)^3}  \dd s\\
        &\leq  \bar C_2 \lambda^2_{\infty} t^2 e^{2 \lambda_\infty t} 
        \sum_{z:\,\lambda^\text{in}(z)\neq 0}
        \left(
        N^2\, e^{- N\frac{\mu^2(z)}{2 \|\mu\|^2_1} e^{-2(r+2)\lambda_\infty t}  }
        + \frac{1 }{4N} \left(\frac{\|\mu\|_1}{\mu(z)}\right)^3 e^{3(r+2)\lambda_{\infty} t}\right)
    \end{align*}
    where the first line follows from \eqref{eq. bounds of the absolution value of hat n_t and tilde n_t} and \eqref{eq. Delta w}, the second line follows from Fubini's theorem, the third line follows from \eqref{eq. second moments of 1/p},
    and the last line follows from \Cref{proposition non-zero of the probability}.
    Finally, applying this and \eqref{eq. convergence proof of FKG-SSA the second part} to \eqref{eq. for comparing the convergence of FKG-SSA}, we prove the result. 
\end{proof}

\subsection{Proving \Cref{them convergence of the FKG-SSA with resampling}} \label{section proof the convergence of the FKG-SSA with resampling}

This subsection is devoted to proving \Cref{them convergence of the FKG-SSA with resampling}, i.e., the convergence of the FKG-SSA when resampling is presented.
First, we introduce a couple of propositions on the estimation of \Cref{algorithm death and division and influx}, which will be used later in the main proof. 

\begin{proposition}\label{proposition estimat of n_t L1 norm}
    Assume \Cref{Assumption 1}---\ref{assumption 2} hold. Let $T$ be the final time in \Cref{algorithm death and division and influx}, $N$ the sample size in this algorithm, and $\hat n_t$ the output of this algorithm.
    Then, there hold
    $$  \|\mu\|_1 \,
    e^{-\lambda_\infty T }
    ~\leq~ \|\hat n_T\|_1 
    ~\leq~ \|\mu\|_1 \,
    e^{\lambda_\infty T } \left(1+ T \|\lambda^\text{in}\|_1\right)
    $$
    with $\lambda_\infty = \max\left\{\|\lambda^\text{react}_1\|_{\infty}, \dots, \|\lambda^\text{react}_r\|_{\infty}, \|\lambda^\text{death}\|_\infty, \|\bar\lambda^\text{div}\|_\infty, \|\lambda^\text{in}\|_\infty \right\}$
    and 
    \begin{align*}
        &\mathbb P
        \left( \frac{\hat n_T(z)+\lambda^{in}(z)/N}{\|\hat n_T+\lambda^{in}/N||_1} 
        \leq  
        \frac{C_0\mu(z)}{2\|\mu\|_1}
        \right)
        \leq \exp\left(-N\,\frac{\mu^2(z)}{2\|\mu\|_1^2} \, e^{-2(r+2)\lambda_\infty T}   \right)
    \end{align*}
    for any $z$ in the support of $\lambda^\text{in}$, where $C_0$ is a positive constant given by
    \begin{align*}
        C_0 = \frac{e^{-(r+3)\lambda_{\infty}T}}{e^{\lambda_\infty T } \left(1+ T \|\lambda^\text{in}\|_1 \right) + \|\lambda^\text{in}\|_1/ \|\mu\|_1}.
    \end{align*}
\end{proposition}

\begin{proof}
    First, we analyze $\|\hat n_T\|$.
    By solving the ODE for $w_i(t)$ in \Cref{algorithm death and division and influx}, we have 
    \begin{align}\label{eq. expression of w_i}
        w_i(t) = e^{\int_{0}^t \lambda^{D}(\texttt{x}_{i}(s)) \dd s}
        + \int_{0}^{t} \left(e^{\int_{s}^t\lambda^{D}(\texttt{x}_{i}(s)) \dd s}\right) \frac{\lambda^\text{in}(\texttt{x}_{i}(s))}{\hat p(t, \texttt{x}_{i}(s))} \dd s.
    \end{align}
    This also suggests that $e^{-\lambda_\infty T}<w_i(t)\leq e^{\lambda_\infty T}\left(1+\int_0^t \frac{\lambda^\text{in}(\texttt{x}_{i}(s))}{\hat p(t, \texttt{x}_{i}(s))} \dd s \right)$.
    Furthermore, by the expression $\hat n_t(x) = \frac{\|\mu\|_1}{N}\sum_{i=1}^{N} \mathbbold 1_{x}(\texttt{x}_{i}(t))w_{i}(t)$ (see \Cref{algorithm death and division and influx}), we can conclude $\|\hat n_t\|_1 = \frac{\|\mu\|_1}{N}\sum_{i=1}^N w_i(t)$, 
    $\|\hat n_T\|_1 \geq \|\mu\|_1 \,
    e^{-\lambda_\infty T }$,
    and
    \begin{align*}
        \|\hat n_T\|_1 
        &\leq \frac{\|\mu\|_1}{N}
        e^{\lambda_\infty T}
        \sum_{i=1}^N \left(1 + \int_0^t \frac{\lambda^\text{in}(\texttt{x}_{i}(s))}{\hat p(t, \texttt{x}_{i}(s))} \dd s \right)\\
        &= 
         \frac{\|\mu\|_1}{N}
        e^{\lambda_\infty T}
        \left(N + \int_0^t \sum_{z:\,\lambda^\text{in}(z)\neq 0 \,\&\, \hat p(s,z)>0 }
        N \lambda^\text{in}(z)
        \dd s \right) \\
        &\leq \|\mu\|_1\,
        e^{\lambda_\infty T}
        \left(1+T \|\lambda^\text{in}\|_1
        \right).
    \end{align*}
    
    Next, we analyze $\frac{\hat n_T(z)+\lambda^{in}(z)/N}{\|\hat n_T+\lambda^{in}/N||_1}$.
    Formula \eqref{eq. expression of w_i} suggests that $w_i(T) \geq e^{-\lambda_{\infty}T}$.
    Combining this with the definition of $\hat n_t(x)$, we can conclude that 
    $\hat n_T(z) \geq \|\mu\|_1 \hat p(T,z) e^{-\lambda_\infty T}$  for any $z$ in the support of $\lambda^\text{in}$ and furthermore
    $\left\{ \hat n_T(z) + \lambda^\text{in}(z)/N \leq  \frac{\mu(z)}{2} e^{-(r+3)\lambda_{\infty}T} \right\}
        \subset 
        \left\{
        \hat p(T,z)\leq \frac{\mu(z)}{2\|\mu\|_1}
        e^{-(r+2)\lambda_{\infty}T} 
    \right\}$.
    This together with the upper bound of $\|n_T\|_1$ suggests 
    \begin{align*}
        \left\{ \frac{\hat n_T(z)+\lambda^{in}(z)/N}{\|\hat n_T+\lambda^{in}/N||_1} \leq  \frac{C_0\mu(z)}{2\|\mu\|_1} \right\}
        &\subset \left\{ {\hat n_T(z)+\lambda^{in}(z)/N}\leq  \frac{\mu(z)}{2} e^{-(r+3)\lambda_{\infty}T} \right\}\\
        &\subset 
        \left\{
        \hat p(T,z)\leq \frac{\mu(z)}{2\|\mu\|_1}
        e^{-(r+2)\lambda_{\infty}T} 
        \right\}\\
        &\subset 
        \left\{
        \hat p(T,z)\leq \frac{1}{2}
        \tilde p(T,z) 
    \right\}.
    \qquad\qquad\quad
    \text{(\Cref{proposition non-zero of the probability})}
    \end{align*}
    Finally, combing this with Hoeffding's inequality, we have 
    \begin{align*}
        \mathbb P
        \left( \frac{\hat n_T(z)+\lambda^{in}(z)/N}{\|\hat n_T+\lambda^{in}/N||_1} \leq  \frac{C_0\mu(z)}{2\|\mu\|_1} \right)
        &
        \leq \mathbb P\left(
        \hat p(T,z)\leq \frac{1}{2}
        \tilde p(T,z) 
        \right) 
        \\
        & \leq 
        e^{-\frac{1}{2} \tilde p^2(T,z) N } 
        \qquad\qquad\qquad\quad
        \text{(Hoeffding's inequality)}
        \\
        & \leq \exp\left(-\frac{\mu^2(z)}{2\|\mu\|_1^2}e^{-2(r+2)\lambda_\infty T} N  \right)
        \quad
        \text{(by \Cref{proposition non-zero of the probability})}
    \end{align*}
    for any $z$ in the support of $\lambda^\text{in}$.
\end{proof}

\begin{proposition}\label{proposition bounds of the hat n}
    Assume \Cref{Assumption 1}---\ref{assumption 2} hold. Let $T$ be the final time in \Cref{algorithm death and division and influx}, $N$ the sample size in this algorithm, $\hat n_t$ the output of this algorithm, and $C_0$ the constants defined in \Cref{proposition estimat of n_t L1 norm}.
    Then, for any $z$ in the support of $\lambda^\text{in}$ and positive constant $C$, there hold 
    \begin{itemize}
        \item $\mathbb E\left[\exp\left\{-C \left(\frac{\hat n_T(z)+\lambda^\text{in}(z)/N}{\|\hat n_T + \lambda^\text{in}/N\|_1}\right)^2\right\}\right]
        \leq 
        \exp\left(-N \,\frac{\mu^2(z)}{2\|\mu\|_1^2}\,e^{-2(r+2)\lambda_\infty T}\right)
        + e^{-C\left(\frac{C_0\mu(z)}{2\|\mu\|_1}\right)^2}
        $
        \item $\mathbb E\left[ \left(\frac{\hat n_T(z)+\lambda^\text{in}(z)/N}{\|\hat n_T + \lambda^\text{in}/N\|_1}\right)^{-\ell}\right]\leq 
        \left(\frac{e^{-(r+3)\lambda_{\infty}T}\|\mu\|_1}{C_0\,\lambda^\text{in}(z)}\right)^\ell
        N^{\ell}
        \exp\left(-N \,\frac{\mu^2(z)}{2\|\mu\|_1^2}\,e^{-2(r+2)\lambda_\infty T}\right)
        + \left(\frac{C_0\mu(z)}{2\|\mu\|_1}\right)^{-\ell}$
        for $\ell = 2, 3$.
    \end{itemize}
    
\end{proposition}

\begin{proof}
    We denote a random variable $q_z \triangleq \frac{\hat n_T(z)+\lambda^\text{in}(z)/N}{\|\hat n_T + \lambda^\text{in}/N\|_1}$ for any $z$ in the support of $\lambda^\text{in}$.
    Then, the first result can be proven by
    \begin{align*}
        \mathbb E\left[e^{-C q_z^2}
        \right]
        &\leq
        \mathbb E\left[
        \mathbbold 1\left( q_z \leq \frac{C_0\mu(z)}{2\|\mu\|_1} \right)
        e^{-C q_z^2}
        \right] 
        + 
        \mathbb E\left[
        \mathbbold 1\left( q_z > \frac{C_0\mu(z)}{2\|\mu\|_1} \right)
        e^{-C q_z^2}
        \right] \\
        &\leq \mathbb P\left(q_z\leq \frac{C_0\mu(z)}{2\|\mu\|_1}\right)
        + e^{-C\left(\frac{C_0\mu(z)}{2\|\mu\|_1}\right)^2} \\
        &\leq \exp\left(-N \,\frac{\mu^2(z)}{2\|\mu\|_1^2}\,e^{-2(r+2)\lambda_\infty T}\right)
        + e^{-C\left(\frac{C_0\mu(z)}{2\|\mu\|_1}\right)^2}
    \end{align*}
    where the last line follows from 
    \Cref{proposition estimat of n_t L1 norm}.
    For $\ell = $ $2$ or $3$, we similarly have 
    \begin{align*}
        &\mathbb E\left[ q_z^{-\ell}\right] 
        \leq 
        \mathbb E\left[ 
        \mathbbold 1\left( q_z\leq \frac{C_0\mu(z)}{2\|\mu\|_1} \right)
        q_z^{-\ell}\right]  
        + 
        \mathbb E\left[ 
        \mathbbold 1\left( q_z > \frac{C_0\mu(z)}{2\|\mu\|_1} \right)
        q_z^{-\ell}\right] \\
        & \leq 
        \left(\frac{\lambda^\text{in}(z)/N}{\|\hat n_T\| + \|\lambda^\text{in}\|}\right)^{-\ell}
        \mathbb P\left(
            q_z \leq \frac{C_0\mu(z)}{2\|\mu\|_1}
        \right)
        +\left(\frac{C_0\mu(z)}{2\|\mu\|_1}\right)^{-\ell}\\
        & \leq
        \left(\frac{\|\mu\|_1 \,
    e^{\lambda_\infty T } \left(1+ T \|\lambda^\text{in}\|_1 \right)+ \|\lambda^\text{in}\|_1}{\lambda^\text{in}(z)}\right)^\ell
        N^{\ell}
        \exp\left(-N \,\frac{\mu^2(z)}{2\|\mu\|_1^2}\,e^{-2(r+2)\lambda_\infty T}\right)
        +\left(\frac{C_0\mu(z)}{2\|\mu\|_1}\right)^{-\ell}
    \end{align*}
    where the last line follows from the upper bounds of $\|\hat n_T\|_{1}$ (see \Cref{proposition bounds of the hat n}).
    This proves the second result. 
\end{proof}

With these propositions, we can now prove \Cref{them convergence of the FKG-SSA with resampling}.

\begin{proof}[\textbf{Proof of \Cref{them convergence of the FKG-SSA with resampling}}]

First, we introduce the scheme of this proof. 
Here, we denote $n_t$ as the unique solution of \eqref{eq n_t dynamics}, $\hat n_t$ as the output of \Cref{algorithm resampling}, and $N$ as the sample size in \Cref{algorithm resampling}.
In addition, we introduce an auxiliary process $\bar n_t$, which is defined as $\bar n_t = n_t$ in the time interval $t\in[0,t_1]$. 
In other time periods $(t_k, t_{k+1}]$ (with $k \in \left\{1, \dots, M \right\}$), $\bar n_t$ is defined as the unique solution of the differential equation
\begin{align}\label{eq. definition of bar n}
    \dot u_t
    = \left(A_\text{react} - \lambda^\text{death} + A_\text{div}\right) u_t + \lambda^\text{in}   \quad \text{for $t\in (t_k, t_{k+1}]$}
    && \text{and}
    && u_{t_k} = \hat n_{t_k} + \lambda^\text{in}/N.
\end{align}
The well-definiteness of $\bar n_t$ is guaranteed by \Cref{proposition uniqueness of the mean dynamics} under the assumed conditions. 
Then, we can decompose the error of $\hat n_t$ by 
\begin{align}\label{eq error decomposition for the algorithm with resampling}
    \mathbb E
    \left[\|\hat n_t - n_t\|_2^2\right]
    ~\leq~ 
    2\,\mathbb E\left[\|\hat n_t - \bar n_t\|_2^2\right]
    + 
    2\,\mathbb E\left[
        \|\bar n_t - n_t\|_2^2
    \right].
\end{align}
In the following, we prove the theorem by analyzing the two terms on the right-hand side of this inequality. 
Specifically, we use mathematical induction to prove the following claims.
\begin{enumerate}[label=\textit{Claim \arabic*}, leftmargin=.8in]
    \item \label{claim the error of bar n}
    For any $t\in[0,T]$, there exist positive constants $\bar C_{t,1}$ and $\bar C_{t,2}$ such that  
        $$\mathbb E\left[\|\bar n_t -  n_t\|_2^2\right]\leq \bar C_{t,1} N^2e^{-\bar C_{t,2}N}
    + \frac{\bar C_{t,1}}{N}.$$
    \item \label{claim the error of hat n} 
    For any $t\in[0,T]$, there exist positive constants $\bar C_{t,1}$ and $\bar C_{t,2}$ such that 
    $$\mathbb E\left[\|\hat n_t - \bar n_t\|_2^2\right]\leq \bar C_{t,1} N^2 e^{-\bar C_{t,2}N}
    + \frac{\bar C_{t,1}}{N}.$$
    \item \label{claim boundedness of the L1 norm}
    For any $k\in\{1,\dots, M+1\}$, there exist positive constants $C^\text{U}_k$ and $C^\text{L}_k$ such that $C^\text{L}_{k} \leq \|\hat n_{t_k}\|_1 \leq C^\text{U}_{k}$.
    \item \label{claim expectation of the exponential term}
    For any $k\in\{1,\dots, M+1\}$, there exist positive constants $\bar C_{k,1}$ and $\bar C_{k,2}$ such that
    \begin{align*}
        \mathbb E\left[\exp\left\{-C \left(\frac{\hat n_{t_k}(z)+\lambda^\text{in}(z)/N}{\|\hat n_{t_k} + \lambda^\text{in}/N\|_1}\right)^2\right\}\right]
        \leq 
        \bar C_{k,1} e^{-\bar C_{k,2} N}
        + e^{-C \bar C_{k,2}}
    \end{align*}
    for any $z$ in the support of $\lambda^\text{in}$ and any positive constant $C$.
    \item \label{claim expectation of 1/p term}
    For any $k\in\{1,\dots, M+1\}$, there exist positive constants $\bar C_{k,1}$ and $\bar C_{k,2}$ such that
    \begin{align*}
        \mathbb E\left[ \left(\frac{\hat n_{t_k}(z)+\lambda^\text{in}(z)/N}{\|\hat n_{t_k} + \lambda^\text{in}/N\|_1}\right)^{-\ell}\right]\leq 
        \bar C_{k,1} N^{\ell} e^{-\bar C_{k,2} N}
        + \bar C_{k,1}
    \end{align*}
    for any $z$ in the support of $\lambda^\text{in}$ and $\ell = $ $2$ or $3$.
\end{enumerate}
Once they are proven, the theorem follows immediately from \ref{claim the error of bar n}, \ref{claim the error of hat n}, and \eqref{eq error decomposition for the algorithm with resampling}.

Next, we prove all the claims using mathematical induction. 
First, for any $t\in[0,t_1]$, there holds $\bar n_t = n_t$ by definition, and, thus, \ref{claim the error of bar n} automatically holds.
Since in this region $\hat n_t$ is the output of \Cref{algorithm death and division and influx} with initial condition $\mu$, \Cref{them convergence of the FKG-SSA} directly proves \ref{claim the error of hat n} for $t\in [0,t_1]$.
Moreover, for $k=1$, \ref{claim boundedness of the L1 norm} follows immediately from \Cref{proposition estimat of n_t L1 norm}, and 
\ref{claim expectation of the exponential term} and \ref{claim expectation of 1/p term} follow from \Cref{proposition bounds of the hat n} and finiteness of the support of $\lambda^\text{in}$ (\Cref{assumption 2}). 

Then, we show that \ref{claim the error of bar n} and \ref{claim the error of hat n} are valid in the interval $[t_{k}, t_{k+1}]$, and that the remaining claims hold for index $k+1$, provided that \ref{claim the error of bar n} and \ref{claim the error of hat n} hold for $t\in[0, t_k]$ and the other claims are true for index $k$.

For \ref{claim the error of bar n}, we can first conclude $\frac{\dd}{\dd t} \left(\bar n_t - n_t\right)
= \left(A_\text{react} - \lambda^\text{death} + A_\text{div}\right) \left(\bar n_t - n_t\right)
$
for $t\in (t_k, t_{k+1}]$ according to \eqref{eq n_t dynamics} and \eqref{eq. definition of bar n}.
Notice that the operator $\left(A_\text{react} - \lambda^\text{death} + A_\text{div}\right) $ is bounded due to the boundedness of the rate functions (\Cref{Assumption 1}). 
By denoting the $L^1$-norm of this operator by $C_\text{op}$, we have 
\begin{align*}
   \|\bar n_t - n_t\|_1
   &= \left \| e^{ \left(A_\text{react} - \lambda^\text{death} + A_\text{div}\right) (t-t_k)} \left(\hat n_{t_k} + \lambda^\text{in}/N-n_{t_k} \right) \right\|_1 \\
   &\leq e^{C_\text{op} (t-t_k)}
   \left(\| \hat n_{t_k} -n_{t_k}\|_1
   + \frac{\|\lambda^\text{in}\|_1}{N}\right)
   && \forall t\in (t_k, t_{k+1}].
\end{align*}
and, furthermore,
\begin{align*}
    \mathbb E\left[
        \|\bar n_t - n_t\|^2_1
    \right]
    &\leq e^{2 C_\text{op} (t-t_k)}
    \left(
    2\,\mathbb E\left[
        \|\hat n_{t_k} - n_{t_k}\|^2_1
    \right]
    + \frac{2\|\lambda^\text{in}\|_1^2}{N^2}
    \right) 
    && \forall t\in (t_k, t_{k+1}].
\end{align*}
Combining this with \ref{claim the error of hat n} for $t=t_k$, we prove \ref{claim the error of bar n} for $t\in (t_k, t_{k+1}]$.

For \ref{claim the error of hat n}, we first notice that $\hat n_t$ (for $t \in (t_k, t_{k+1}]$) is the output of \Cref{algorithm death and division and influx} with the initial condition $\hat n_{t_k} + \lambda^\text{in}/N$, and $\bar n_t$ is the solution of \eqref{eq. definition of bar n} with the same initial condition. 
Also, the support of this initial condition contains the support of $\lambda^\text{in}$.
Therefore, by \Cref{them convergence of the FKG-SSA}, we can find positive $\tilde C_{t,1}$ and $\tilde C_{t,2}$ for any given $t\in(t_k, t_{k+1}]$ such that 
\begin{align*}
    &\mathbb E\left[\|\hat n_t - \bar n_t\|_2^2 \,\big|\, \hat n_{t_k} \right] \leq 
    \tilde C_{t,1} N^2 \left(
    \sum_{z:\,\lambda^\text{in}(z)\neq 0}
    e^{-\tilde C_{t,2}\,N \,\left(\frac{\hat n_{t_k}(z) + \lambda^\text{in}(z)/N}{\|\hat n_{t_k} + \lambda^\text{in}/N\|_1}\right)^2}\right) \\
    & ~~\qquad + \frac{\tilde  C_{t,1}}{N}
    \left(
    1 + \sum_{z:\,\lambda^\text{in}(z)\neq 0} 
    \left(\frac{\hat n_{t_k}(z) + \lambda^\text{in}(z)/N}{\|\hat n_{t_k} + \lambda^\text{in}/N\|_1}\right)^{-2}
    + \sum_{z:\,\lambda^\text{in}(z)\neq 0} 
    \left(\frac{\hat n_{t_k}(z) + \lambda^\text{in}(z)/N}{\|\hat n_{t_k} + \lambda^\text{in}/N\|_1}\right)^{-3}
    \right).
\end{align*}
By combining this with \ref{claim expectation of the exponential term} and \ref{claim expectation of 1/p term} (for index $k$), we can further find positive constants $\hat C_{k,1}$ and $\hat C_{k,2}$ such that 
\begin{align*}
    &\mathbb E\left[\|\hat n_t - \bar n_t\|_2^2 \right]\\
    &\leq \tilde C_{t,1} N^{2}
    \left( \sum_{z:\,\lambda^\text{in}(z)\neq 0} \hat C_{k,1} e^{-\hat C_{k,2} N }
    + e^{-\hat C_{k,2}\tilde C_{t,2} N}
    \right)\\
    &\quad + \frac{\tilde C_{t,1}}{N}
    \left[
        1
        + 
        \sum_{z:\,\lambda^\text{in}(z)\neq 0}
        \left(
            \hat C_{k,1} N^2 e^{-\hat C_{k,2} N} + \hat C_{k,1}
        \right)
        +
        \sum_{z:\,\lambda^\text{in}(z)\neq 0}
        \left(
            \hat C_{k,1} N^3 e^{-\hat C_{k,2} N} + \hat C_{k,1}
        \right)
    \right]\\
    &\leq  
    \tilde C_{t,1} \left(\sum_{z:\,\lambda^\text{in}(z)\neq 0} \left(1+ 3 \hat C_{k,1}\right)\right)
    N^{2} e^{- \min\{\hat C_{k,2}\,,\,\hat C_{k,2}\tilde C_{t,2}  \}N}
    + \frac{\tilde C_{t,1}\left(1 + 2\sum_{z:\,\lambda^\text{in}(z)\neq 0}
    \hat C_{k,1}
    \right)}{N}
\end{align*}
where the last line follows from the relation $N<N^2$.
This proves \ref{claim the error of hat n} for $t\in(t_{k},t_{k+1}]$.

For \ref{claim boundedness of the L1 norm}, we can conclude by \Cref{proposition estimat of n_t L1 norm} that 
\begin{align*}
     \left\|\hat n_{t_k} + \frac{\lambda^\text{in}}{N}\right\|_1 \,
    e^{-\lambda_\infty (t_{k+1}-t_{k}) }
    ~\leq~ \|\hat n_{t_{k+1}}\|_1 
    ~\leq~ \left\|\hat n_{t_k} + \frac{\lambda^\text{in}}{N} \right\|_1  \,
    e^{\lambda_\infty (t_{k+1}-t_{k})} \left[1+ (t_{k+1}-t_{k} )\|\lambda^\text{in}\|_1 \right].
\end{align*}
Combining this with \ref{claim boundedness of the L1 norm} (for index $k$), we can conclude 
\begin{align*}
    C^{L}_k e^{-\lambda_\infty (t_{k+1}-t_{k}) }
    ~\leq~ \|\hat n_{t_{k+1}}\|_1 
    ~\leq~
    \left(
        C^{U}_k + \|\lambda^\text{in}\|_1
    \right)
    e^{\lambda_\infty (t_{k+1}-t_{k})} \left[1+ (t_{k+1}-t_{k} )\|\lambda^\text{in}\|_1 \right]
\end{align*}
which proves \ref{claim boundedness of the L1 norm} for index $k+1$.

For \ref{claim expectation of the exponential term}, we first define a positive random variable $\tilde C_0$ by
\begin{align}\label{eq. definition of tilde C_0}
   \tilde C_0
    \triangleq 
    \frac{e^{-(r+3)\lambda_{\infty}(t_{k+1}-t_k)}}{e^{\lambda_\infty (t_{k+1}-t_k) } \left[1+ (t_{k+1}-t_k) \|\lambda^\text{in}\|_1 \right] + \frac{\|\lambda^\text{in}\|_1}{ \|\hat n_{t_k} + \lambda^\text{in}/N\|_1}}.
\end{align}
According to \eqref{claim boundedness of the L1 norm}, the quantity $ \|\hat n_{t_k} + \lambda^\text{in}/N\|_1$ is bounded both above and below, with its lower bound strictly greater than zero.
Therefore, the random variable $\tilde C_0$ is lower bounded.
We denote its lower bound by $\tilde C_0^{L}$. 
Then, according to \Cref{proposition bounds of the hat n}, we can find positive constant $\hat C$ such that 
\begin{align*}
    \mathbb E\left[\left. e^{-C \left(\frac{\hat n_{t_{k+1}}(z)+\lambda^\text{in}(z)/N}{\|\hat n_{t_{k+1}} + \lambda^\text{in}/N\|_1}\right)^2}
    \right| \hat n_{t_k}
    \right]
    &\leq 
    e^{-\hat C \left(\frac{\hat n_{t_{k}}(z)+\lambda^\text{in}(z)/N}{\|\hat n_{t_{k}} + \lambda^\text{in}/N\|_1}\right)^2 N }
    + e^{-C \left(\frac{\tilde C_0 \left(\hat n_{t_{k}}(z)+\lambda^\text{in}(z)/N\right)}{2 \|\hat n_{t_{k}} + \lambda^\text{in}/N\|_1}\right)^2 }\\
    & \leq 
    e^{-\hat C \left(\frac{\hat n_{t_{k}}(z)+\lambda^\text{in}(z)/N}{\|\hat n_{t_{k}} + \lambda^\text{in}/N\|_1}\right)^2 N }
    + e^{-\frac{C \left(\tilde C^{L}_0\right)^2}{4} \left(\frac{\hat n_{t_{k}}(z)+\lambda^\text{in}(z)/N}{\|\hat n_{t_{k}} + \lambda^\text{in}/N\|_1}\right)^2 }
\end{align*}
for any $z$ in the support of $\lambda^\text{in}$ and any positive constant $C$.
Then, by applying \ref{claim expectation of the exponential term} (with index $k$) to this inequality,
we further have 
\begin{align*}
    \mathbb E\left[e^{-C \left(\frac{\hat n_{t_{k+1}}(z)+\lambda^\text{in}(z)/N}{\|\hat n_{t_{k+1}} + \lambda^\text{in}/N\|_1}\right)^2}
    \right]
    &\leq 
    \bar C_{k,1} e^{-\bar C_{k,2}N}
    + e^{-\bar C_{k,2} \hat C N}
    + \bar C_{k,1} e^{-\bar C_{k,2}N}
    + e^{-\bar C_{k,2}\frac{C \left(\tilde C_0^L\right)^2}{4} }\\
    &\leq 
    \left(2 \bar C_{k,1} +1 \right)
    e^{-\min\{\bar C_{k,2} \,,\,\bar C_{k,2} \hat C\} N }
    + e^{-C \frac{\bar C_{k,2} \left(\tilde C_0^L\right)^2}{4} }
\end{align*}
where $\bar C_{k,1}$ and $\bar C_{k,2}$ are the constants given in \ref{claim expectation of the exponential term}.
This proves \ref{claim expectation of the exponential term} for index $k+1$.

For \ref{claim expectation of 1/p term}, we can prove it similarly as for \ref{claim expectation of the exponential term}.
\Cref{proposition bounds of the hat n} and the finiteness of the support of $\lambda^\text{in}$ suggest that there exist positive constants $\tilde C'_1$ and $\tilde C_2$ such that 
\begin{align*}
    &\mathbb E\left[\left.  \left(\frac{\hat n_{t_{k+1}}(z)+\lambda^\text{in}(z)/N}{\|\hat n_{t_{k+1}} + \lambda^\text{in}/N\|_1}\right)^{-\ell}
    \right| \hat n_{t_k}
    \right]\\
    &\leq 
    \left(\frac{\tilde C'_1 \|n_{t_k}+\lambda^\text{in}/N\|_1}{\tilde C_0}\right)^\ell
    N^{\ell}
    e^{-\tilde C_2 N \left(\frac{\hat n_{t_{k}}(z)+\lambda^\text{in}(z)/N}{\|\hat n_{t_{k}} + \lambda^\text{in}/N\|_1}\right)^2}
    + \left(\frac{2}{\tilde C_0}\right)^{\ell}
    \left(\frac{\hat n_{t_{k+1}}(z)+\lambda^\text{in}(z)/N}{\|\hat n_{t_{k+1}} + \lambda^\text{in}/N\|_1}\right)^{-\ell}\\
    &\leq 
    \left(\frac{\tilde C'_1 \|n_{t_k}+\lambda^\text{in}/N\|_1}{\tilde C_0^{L}}\right)^\ell
    N^{\ell}
    e^{-\tilde C_2 N \left(\frac{\hat n_{t_{k}}(z)+\lambda^\text{in}(z)/N}{\|\hat n_{t_{k}} + \lambda^\text{in}/N\|_1}\right)^2}
    + \left(\frac{2}{\tilde C_0^{L}}\right)^{\ell}
    \left(\frac{\hat n_{t_{k+1}}(z)+\lambda^\text{in}(z)/N}{\|\hat n_{t_{k+1}} + \lambda^\text{in}/N\|_1}\right)^{-\ell}
\end{align*}
for any $z$ in the support of $\lambda^\text{in}$ and $\ell\in\{2,3\}$.
Here, $\tilde C_0$ is defined in \eqref{eq. definition of tilde C_0}, and $\tilde C_0^{L}$ is its lower bound, as discussed in the previous paragraph. 

Then, by applying \ref{claim expectation of the exponential term} and \ref{claim expectation of 1/p term} to this inequality, we can further find constants $\tilde C_{k,1}$ and $\tilde C_{k,2}$ such that 
\begin{align*}
    &\mathbb E\left[ \left(\frac{\hat n_{t_{k+1}}(z)+\lambda^\text{in}(z)/N}{\|\hat n_{t_{k+1}} + \lambda^\text{in}/N\|_1}\right)^{-\ell}
    \right]\\
    &\leq \tilde C_1 N^{\ell}\left(\tilde C_{k,1}\, e^{-\tilde C_{k,2} N} + e^{-\tilde C_2 \tilde C_{k,2} N}\right)
     + \tilde C_1 \left(
    \tilde C_{k,1}N^{\ell}\, e^{-\tilde C_{k,2} N} + \tilde C_{k,1}
    \right)\\
    & \leq 
    \tilde C_1  \left(2\tilde C_{k,1}+1\right)N^{\ell} e^{-\min\{\tilde C_{k,2}\,,\, \tilde C_2 \tilde C_{k,2}\}N}
    + \tilde C_1 \tilde C_{k,1}
\end{align*}
for any $z$ in the support of $\lambda^\text{in}$ and $\ell\in\{2,3\}$.
This proves \ref{claim expectation of 1/p term} for index $k+1$.

By combining all these results, we prove all the claims in the entire time interval $[0,T]$ and for all the index $k\in\{0, 1, \dots, M\}.$
The theorem then follows immediately from \ref{claim the error of bar n}, \ref{claim the error of hat n}, and \eqref{eq error decomposition for the algorithm with resampling}
\end{proof}
\bibliographystyle{siamplain}
\bibliography{references}

\end{document}